\documentclass{marticle}
\usepackage[utf8]{inputenc}
\title{Improved Analysis of Higher Order Random Walks and Applications}
\author{Vedat Levi Alev \thanks{Supported by the David R.~Cheriton Graduate
Scholarship and the NSERC Discovery Grant 2950-120715. E-mail: \href{mailto:vlalev@uwaterloo.ca}{vlalev@uwaterloo.ca}}
\and Lap Chi Lau\thanks{Supported by NSERC Discovery Grant 2950-120715. E-mail: \href{mailto:lapchi@uwaterloo.ca}{lapchi@uwaterloo.ca}}
}
\date{}

\begin{document}

\begin{titlepage}
    \def\thepage{}
    \thispagestyle{empty}

    \maketitle
    \begin{abstract}
        The motivation of this work is to extend the techniques of higher order random walks on simplicial complexes to analyze mixing times of Markov chains for combinatorial problems.
        Our main result is a sharp upper bound on the second eigenvalue of the down-up walk on a pure simplicial complex, in terms of the second eigenvalues of its links.
        We show some applications of this result in analyzing mixing times of Markov chains, 
        including sampling independent sets of a graph and sampling common independent sets of two partition matroids.
    \end{abstract}

\end{titlepage}

\newpage

\section{Introduction}\label{sec:intro}

Consider the following random walks~\cite{KaufmanM17,DinurK17,KaufmanO18,DiksteinDFH18} defined\footnote{All the
definitions in the introduction will be formally defined again in a more general setting in \cref{sec:prelim}.} on a simplicial complex $X$. 
Initially, the random walk starts from an arbitrary face $\alpha_1$ of
dimension $k$ in $X$. 
\begin{itemize}
    \item {\sc Down-Up Walk}: In each step $t \geq 1$, we choose a uniform random
        element $i \in \alpha_t$ and delete $i$ from $\alpha_t$, and set
        $\alpha_{t+1}$ to be a uniform random face of dimension $k$ in $X$ that
        contains $\alpha_t\setminus \{i\}$.
        This is called the $k$-th down-up walk of $X$, and its transition matrix is denoted by $\DownW_k$.
    \item {\sc Up-Down Walk}:
        In each step $t \geq 1$, we choose a uniform random face $\beta$ of dimension
        $k+1$ in $X$ that contains $\alpha_t$, and choose a uniform random
        element $i \in \beta$ and set $\alpha_{t+1} = \beta \setminus \{i\}$.
        This is called the $k$-th up-down walk of $X$, and its transition matrix is denoted by $\UpW_k$.
\end{itemize}
The stationary distribution of these random walks is the uniform distribution on the faces of dimension $k$ in the simplicial complex $X$. 
The question of interest is the mixing time of these random walks, i.e.~the number of steps $t$ required for the distribution of $\alpha_t$ to be close to the uniform distribution.

A graph is a simplicial complex of dimension $1$.
The transition matrix of the lazy random walk on a graph is $\UpW_0$.
Fundamental results in spectral graph theory state that (i) the mixing time of
the lazy random walk is small, if and only if (ii) the second eigenvalue of $\UpW_0$ is small, if and only if
(iii) the graph is an expander graph.
See~\cite{HooryLW06, WilmerLP09} for surveys on this topic.

Since the theory of expander graphs has many applications, there are various
motivations in generalizing these results for graphs to simplicial complexes. 
Several definitions of high-dimensional expanders have been studied in the literature (e.g.~\cite{LinialM06,Gromov10,ParzanchevskiRT16,DotterrerKW16,KaufmanM17,Oppenheim18}),
and these results have found interesting applications in discrete geometry, 
complexity theory, coding theory, and property testing (e.g.~\cite{LinialM06,MeshulamW09, FoxGLNP11, KaufmanL14, EvraK16, KaufmanM16a, KaufmanKL16, DinurK17,DinurHKNT19}).

\subsubsection*{Local Spectral Expanders}

In this paper, we consider the definition of $\gamma$-local-spectral expanders developed in~\cite{KaufmanM17,DinurK17,KaufmanO18,Oppenheim18,DiksteinDFH18} for the study of random walks on simplicial complexes.
The local structures of a simplicial complex are described by its links.
The link $X_{\alpha}$ of a face $\alpha \in X$ is defined as the simplicial
complex $X_\alpha = \{\beta \setminus \alpha: \beta \in X, \beta \supset \alpha\}$.
The graph $G_{\alpha}=(V_{\alpha},E_{\alpha})$ of the link $X_{\alpha}$ is defined as follows: (i) each vertex $i$ in $V_{\alpha}$ corresponds to a singleton
$\{i\}$ in $X_{\alpha}$, (ii) two vertices $i,j \in V_{\alpha}$ have an edge in
$E_{\alpha}$ if and only if $\set{i,j}$ is contained in some face of
$X_{\alpha}$, (iii) the weight $w_{ij}$ of an edge $ij \in E_{\alpha}$ is
proportional to the number of maximal faces in $X_{\alpha}$ that contains
$\set{i,j}$.

Informally, a simplicial complex $X$ is a $\gamma$-local-spectral expander
if $G_{\alpha}$ is an expander graph for every $\alpha\in X$. 
In the following, we say $X$ is a pure simplicial complex if every maximal face of $X$ is of the same dimension, and we call this the dimension of $X$.
\begin{defn}[$\gamma$-local-spectral expanders~\cite{Oppenheim18,KaufmanO18}] \label{d:local}
    A $d$-dimensional pure simplicial complex $X$ is a $\gamma$-local-spectral expander if $\lambda_2(G_{\alpha}) \leq \gamma$ for every face $\alpha \in X$ of dimension up to $d-2$, 
    where $\lambda_2(G_{\alpha})$ denotes the second largest eigenvalue
    of the random walk matrix of $G_{\alpha}$ (where the transition probabilities are proportional to the edge weights).
\end{defn}

\subsubsection*{Kaufman-Oppenheim Theorem}

Kaufman and Oppenheim~\cite{KaufmanO18} proved that the $k$-th down-up walk
and the $(k-1)$-th up-down walk have a non-trivial spectral gap as long as the
simplicial complex is a $\gamma$-local-spectral expander for $\gamma < 2/k^2$.  

\begin{theorem}[\cite{KaufmanO18}]\label{t:KO}
    Let $X$ be a pure $d$-dimensional simplicial complex.
    Suppose $X$ is a $\gamma$-local-spectral expander. 
    Then, for every $0 \leq k \leq d$,
    \[ 
    \lambda_2(\DownW_{k}) =\lambda_2(\UpW_{k-1}) \le 1 - \frac{1}{k+1} +
    \frac{k\gamma}{2},
    \]
\end{theorem}
\cref{t:KO} states that the spectral gap of $\DownW_{k}$ is at least
$g:=1-\lambda_2(\DownW_{k}) \geq \frac{1}{k+1} - \frac{k \gamma}{2}$, which implies by a standard argument (see~\cref{thm:spec-mix-bd}) that the mixing time of these walks is at most $O(\frac{(k+1)\log(n)}{g})$ where $n$ is the size of the ground set of $X$.
For example, if $\gamma \leq 0$, then the mixing time of $\DownW_{k}$ is at
most $O(k^2 \log(n))$.

\cref{t:KO} can also be used to bound the spectral gap of certain ``longer''
random walks on simplicial complexes (see \cref{cor:m1m2dk} and
\cref{ss:long}). 
Dinur and Kaufman~\cite{DinurK17} use these results with the Ramanujan
complexes of \cite{LubotzkySV05} to construct efficient agreement testers, which have applications to PCP constructions.
Recently, these ideas have also found applications in coding theory~\cite{DinurHKNT19}.

\subsubsection*{Oppenheim's Trickling Down Theorem}

Kaufman-Oppenheim \cref{t:KO} provides a way to bound the mixing time of the down-up walks and up-down walks.
To apply the theorem, however, one needs to check that $\lambda_2(G_{\alpha})
\leq \gamma$ for every face $\alpha \in X$ of dimension at most $d-2$.
This is not an easy task.
There are exponentially many graphs $G_{\alpha}$ to check, 
and these graphs are defined implicitly where computing the edge weights involve non-trivial counting problems.
A very useful result by Oppenheim~\cite{Oppenheim18} makes this task easier, 
by relating the second eigenvalue of the graph of a lower-dimensional link to that of a higher-dimensional link.

\begin{theorem}[\cite{Oppenheim18}] \label{t:opp}
    Let $X$ be a pure $d$-dimensional simplicial complex.
    Suppose $\lambda_2(G_{\beta}) \leq \gamma\le \frac{1}{2}$  for every face
    $\beta$ of dimension $k$, and $G_{\alpha}$ is connected for every face
    $\alpha$ of dimension $k-1$.
    Then, for every face $\alpha$ of dimension $k-1$, it holds that
    \[
        \lambda_2(G_\alpha) \leq \frac{\gamma}{1-\gamma}.
    \]
\end{theorem}

Applying this theorem inductively, we can reduce the problem of bounding
$\lambda_2(G_{\alpha})$ for every $\alpha$ to bounding $\lambda_2(G_{\beta})$
for only those faces $\beta$ of highest dimension. 

\begin{corollary}[\cite{Oppenheim18}] \label{c:opp}
    Let $X$ be a pure $d$-dimensional simplicial complex.
    Suppose $\lambda_2(G_{\beta}) \leq \gamma \le \frac{1}{d}$ for every face $\beta$ of dimension $d-2$, 
    and $G_{\alpha}$ is connected for every face $\alpha$.
    Then, for every $k \leq d-2$, and for every face $\alpha$ of dimension $k$, it holds that
    \[
        \lambda_2(G_{\alpha}) \leq \frac{\gamma}{1-(d-2-k)\gamma}.
    \]
\end{corollary}

\cref{c:opp} is useful for two reasons:
First, note that the weight of every edge in $G_{\beta}$ for face $\beta$ of dimension $d-2$ is either zero or one,
which makes the task of bounding its second eigenvalue more tractable.
Second, if one can prove that $\lambda_2(G_{\beta}) = O(\frac{1}{d^2})$ for every face $\beta$ of dimension $d-2$ and $G_\alpha$ is connected for every face $\alpha$, then one can conclude that $\lambda_2(G_{\alpha}) = O(\frac{1}{d^2})$ for every face $\alpha$ and hence the simplicial complex is a $O(\frac{1}{d^2})$-local-spectral expander.
So, the reduction of Oppenheim is basically lossless in the regime where
Kaufman-Oppenheim's \cref{t:KO} applies.

\subsubsection*{Analyzing Mixing Times of Markov Chains}

Recently, Anari, Liu, Oveis Gharan, and Vinzant~\cite{AnariLOV18} found a
striking application of \cref{t:KO} and \cref{c:opp} in proving
the matroid expansion conjecture of Mihail and Vazirani \cite{MihailV87},
answering a long standing open question in Markov chain Monte Carlo
methods.

To illustrate their result, consider the special case of sampling a random spanning tree from a graph $G=(V,E)$.
Let $X$ be the simplicial complex where the ground set is $E$ and each acyclic subgraph of $G$ is a face of $X$.
Then $X$ is a pure $d$-dimensional simplicial complex, where $d = |V|-2$ and the spanning trees of $G$ are the maximal faces of $X$.
Note that $\DownW_{d}$ in $X$ is exactly the natural Markov chain on the spanning trees of $G$, where in each step we delete a uniformly random edge $e$ from the current spanning tree $T$ and add a uniformly random edge $f$ so that $T-e+f$ is a spanning tree.
So, the problem of proving the Markov chain on spanning trees is fast mixing is equivalent to upper bounding $\lambda_2(\DownW_d)$ of the simplicial complex $X$.

Using the nice structures of matroids, Anari, Liu, Oveis Gharan, and Vinzant~\cite{AnariLOV18} 
 showed that the graph $G_{\beta}$ is a complete multi-partite graph for every face $\beta$ of dimension $d-2$, and this implies that $\lambda_2(G_{\alpha}) \leq 0$ for every face $\beta$ of dimension $d-2$.
Thus, it follows from Oppenheim's \cref{c:opp} that $\lambda_2(G_{\alpha})
\leq 0$ for every face $\alpha$.\footnote{The result that every matroid complex is a $0$-local-spectral expander was also proved by Huh and Wang~\cite{HuhW17},
using techniques from Hodge theory for matroids~\cite{AdirpasitoHK18} instead of Oppenheim's theorem.}
Then Kaufman-Oppenheim's \cref{t:KO} implies that
$\lambda_2(\DownW_{d}) \leq 1-\frac{1}{d+1}$, 
and thus the mixing time of the
Markov chain of sampling matroid bases is at most $O(d^2 \log n)$.
This provides the first FPRAS for counting the number of matroid bases,
and also proves that the basis exchange graph of a matroid is an expander graph. 

The proof of the matroid expansion conjecture shows that the techniques developed in higher order random walks provide a new simplicial complex approach to analyze mixing times of Markov chains.
It is thus natural to investigate whether this approach can be extended to other problems. 
Here we would like to discuss some limitations of the current techniques.
It can be shown that $\lambda_2(G_{\beta}) \leq 0$ only if $G_{\beta}$ is a
complete multi-partite graph~\cite{GodsilWP} and more generally a
$0$-local-expander is a weighted matroid
complex~\cite{BrandenH19}, and so the same analysis as in~\cite{AnariLOV18} only works for matroids.
Note that Kaufman-Oppenheim \cref{t:KO} only applies when $\lambda_2(G_{\alpha}) \leq O(\frac{1}{d^2})$ for every face $\alpha$ up to dimension $d-2$.
For many problems that we have considered, it does not hold that $\lambda_2(G_{\beta}) \leq O(\frac{1}{d^2})$ even when restricted to faces $\beta$ of dimension $d-2$.

\subsection{Main Result}\label{sec:intromain}

The main motivation of this work is to extend this simplicial complex approach to analyze mixing times of more general Markov chains.
Our main result is the following improved eigenvalue bound for higher order random walks.

\begin{theorem} \label{t:main}
    Let $X$ be a pure $d$-dimensional simplicial complex.
    Define 
    \[\gamma_j := \max_{\alpha} \{ \lambda_2(G_{\alpha}) : \alpha \in X~\textrm{ and }~\alpha {\rm~is~of~dimension~} j \},\] 
    For any $0 \leq k \leq d$,
    \[ \lambda_2(\DownW_{k}) = \lambda_2(\UpW_{k-1}) \le 1 - \frac{1}{k+1} \prod_{j = - 1}^{k-2} (1 - \gamma_j).\]
\end{theorem}

The following are some remarks about \cref{t:main}.
\begin{enumerate}
    \item 
        A basic result is that a simplicial complex $X$ is gallery connected (i.e.~$\lambda_2(\DownW_d) < 1$) if $G_{\alpha}$ is connected (i.e.~$\lambda_2(G_\alpha) < 1$) for every face $\alpha$ of dimension up to $d-2$.  
        \cref{t:main} provides a quantitative generalization of this result.
    \item
        A corollary of \cref{t:main} is that the spectral gap $1-\lambda_2(\DownW_k)$ of the $k$-th down-up walk is at least $\Omega(1/k)$ if $X$ is a $O(\frac{1}{k})$-local-spectral expander.
        This is an improvement of \cref{t:KO} where it requires the simplicial complex $X$ to be a $O(\frac{1}{k^2})$-local-spectral expander to conclude that $\DownW_k$ has a non-zero spectral gap.
    \item
        It can be shown that the spectral gap $1-\lambda_2(\DownW_k)$ of the $k$-th down-up walk is at most $O(\frac{1}{k})$ for any simplicial complex (see \cref{prop:nonexp}),
        so \cref{t:main} shows that any $O(\frac{1}{k})$-local-spectral expander has the optimal spectral gap for the $k$-th down-up walk up to a constant factor.
    \item
        The refinement of having a different bound $\gamma_j$ for links of different dimension is very useful for analyzing Markov chains.  
        We will see some applications in \cref{sec:sampling}.
    \item 
        \cref{t:main} can be used to provide a tighter bound on the spectral gap of
        certain ``longer'' random walks (see \cref{c:fakem1m2})
        which were known to be useful in coding theory and agreement testing (see \cref{ss:long}).
\end{enumerate}

Combined with Oppenheim's \cref{t:opp},
\cref{t:main} provides the following bound for the second eigenvalue of higher order random walks in a black box fashion.  
See~\cref{sec:eigenbd} for the proof.

\begin{corollary} \label{c:convenient}
    Let $X$ be a pure $d$-dimensional simplicial complex.
    For any $0 \leq k \leq d$,
    suppose $\gamma_{k-2} \leq \frac{1}{k+1}$ and $G_{\alpha}$ is connected for every face $\alpha$ up to dimension $k-2$,
    then
    \[
        \lambda_2(\DownW_k) = \lambda_2(\UpW_{k-1}) \leq 1 - \frac{1}{(k+1)^2}.
    \]
\end{corollary}

This provides a convenient way to bound the mixing time of Markov chains.
Recall that the edge weights in $G_{\beta}$ for face $\beta$ of dimension $d-2$ are either zero or one, and so it is easier to bound their second eigenvalue.
\cref{c:convenient} states that as long as we can prove $\lambda_2(G_{\beta}) \leq 1/(d+1)$ for these unweighted graphs in the highest dimension,
then we can conclude that $\DownW_d$ is fast mixing.

\subsection{Applications}

We present several applications of \cref{t:main} and \cref{c:convenient}, in analyzing mixing times of Markov chains (\cref{ss:IS}, \cref{ss:matroid}, \cref{ss:correlation}), in analyzing constructions of high-dimensional expanders (\cref{ss:LMY}), and in analyzing longer random walks (\cref{ss:long}).

\subsubsection{Sampling Independent Sets of Fixed Size} \label{ss:IS}

One of the most natural simplicial complexes to consider is the independent set complex of a graph~\cite{Meshulam06,AharoniB06}.
Let $G=(V,E)$ be a graph.
The independent set complex $I_{G,k}$ has the vertex set $V$ as the ground set, and a subset $S \subset V$ is a face in $X$ if and only if $S$ is an independent set in $G$ with $|S| \leq k$.

We are interested in bounding $\lambda_2(\DownW_{k-1})$ for this simplicial complex $X$.
The $(k-1)$-th down-up walk corresponds to a natural Markov chain on sampling independent sets of size $k$.
Initially, the random walk starts from an arbitrary independent set $S_1$ of size $k$.
In each step $t \geq 1$, we choose a uniform random vertex $u \in S_t$ and delete it from $S_t$, and we choose a uniform random vertex $v$ so that $S_t-u+v$ is still an independent set of size $k$ and set $S_{t+1} := S_t - u + v$.
This Markov chain is known to mix in polynomial time for $k \leq \frac{|V|}{2\Delta+1}$ where
$\Delta$ is the maximum degree of $G$, by using the path coupling technique~\cite{BubleyD97,MitzenmacherUpfal05}.
We prove a more refined result using the simplicial complex approach.

\begin{restatable}{theorem}{IS} \label{t:IS}
    Let $G = (V,E)$ be a graph with maximum degree $\Delta$.
    Let $\DownW_{k-1}$ be the $(k-1)$-th down-up walk on the simplicial complex $I_{G,k}$.
    Let $\Aye_G$ be the adjacency matrix of $G$.
    \[
        {\rm If} \quad k \leq \frac{|V|}{\Delta + |\lambda_{\min}(\Aye_G)|},
    \quad {\rm then} \quad
    \lambda_2(\DownW_{k-1}) \leq 1 - \frac{1}{k^2}.
\]
\end{restatable}

It is well-known that $|\lambda_{\min}(A_G)| \leq \Delta$ for a graph with maximum degree $\Delta$, and so Theorem~\ref{t:IS} recovers the previous result that the Markov chain is fast mixing if $k \leq \frac{|V|}{2\Delta}$.
There are various graph classes with $|\lambda_{\min}(A_G)|$ smaller than $\Delta$, and \cref{t:IS} allows us to sample larger independent sets.
For example, it is known that $|\lambda_{\min}(A_G)| \leq O(\sqrt{\Delta})$ for planar graphs and more generally for graphs with bounded arboricity~\cite{Hayes06}, and also for random graphs and more generally for two-sided expander graphs~\cite{HooryLW06}.

\subsubsection{Sampling Common Independent Sets in Two Partition Matroids} \label{ss:matroid}

A matroid $M = (E, \Ii)$ on the ground set $E$ with the set of independent sets
${\Ii} \subset 2^E$ is a combinatorial object satisfying the following properties:
\begin{itemize}
   \item (containment property) if $S \in \Ii$ and $T \subset S$, then $T \in \Ii$,
    \item (extension property) if $S, T \in \Ii$ such that $|S| > |T|$ then
        there is some $x \in S\backslash T$ such that $\set{x} \cup T \in \Ii$.
\end{itemize}
A partition matroid is the special case where the ground set $E$ is partitioned into disjoint blocks $B_1, \ldots, B_l \subseteq E$ with parameters $0 \leq d_i \leq |B_i|$ for $1 \leq i \leq l$,
and a subset $S$ is in $\Ii$ if and only if $|S \cap B_i| \leq d_i$ for $1 \leq i \leq l$.

The intersection of two matroids $M_1 = (E, \Ii_1)$ and $M_2 = (E, \Ii)$ over the same ground set $E$ can be used to formulate various interesting combinatorial optimization problems~\cite{Schrijver03}. 
We are interested in the problem of sampling a uniform random common
independent set of size $k$, i.e. a random subset $F \in \Ii_1 \cap \Ii_2$ with $|F|=k$.

Matroids naturally correspond to simplicial complexes.
Let $C_{M_1,M_2,k}$ be the matroid intersection complex with ground set $E$, where a subset $F \subset E$ is a face in $C_{M_1,M_2,k}$ if and only if $F \in
{\Ii_1} \cap {\Ii_2}$ and $|F| \leq k$.
The $(k-1)$-th down-up walk of this complex corresponds to a natural Markov chain on sampling common independent sets of $M_1$ and $M_2$ of size $k$.
We show that this Markov chain is fast mixing for $k$ up to one third the size of a maximum common independent set, when $M_1$ and $M_2$ are partition matroids and there are no two elements belonging to the same block in both matroids (i.e. there are no two elements $x,y$ such that $x$ and $y$ are in the same block in $M_1$ and also in the same block in $M_2$).

\begin{restatable}{theorem}{matroid} \label{t:matroid}
Let $M_1 = (E, {\mathcal I_1})$ and $M_2 = (E, {\mathcal I_2})$ be two given partition matroids with a common independent set of size $r$ and no two elements belonging to the same block in both matroids. 
If $k \leq r/3$, then
    \[ \lambda_2(\DownW_{k-1}) \le 1- \frac{1}{k^2},\]
where $\DownW_{k-1}$ is the $(k-1)$-th down-up walk on the matroid intersection complex $C_{M_1, M_2, k}$.
\end{restatable}

The proof of \cref{t:matroid} reveals an interesting property of the links of the simplicial complex $C_{M_1,M_2,k}$.
For any face $\beta$ of dimension $k-3$, we show that the graph $G_{\beta}$ is the complement of the line graph of a bipartite graph.
We note that this holds for any two matroids, not just for partition matroids.
By the additional assumptions that the two matroids are partition matroids and there are no two elements in the same block in both matroids,
the graph $G_{\beta}$ is the line graph of a {\em simple} bipartite graph.
Using the fact that the adjacency matrix of the line graph of a simple graph has minimum eigenvalue
at least $-2$, we prove that $\lambda_2(G_{\beta}) \leq \frac{1}{k}$ as long as
$k \leq \frac{r}{3}$.
We can then use \cref{c:convenient} to conclude \cref{t:matroid}.

\subsubsection{Sampling Independent Sets from Hardcore Distributions} \label{ss:correlation}

Very recently, Anari, Liu, and Oveis Gharan~\cite{AnariLO19} use \cref{t:main} to prove a strong result about sampling independent sets from the hardcore distribution.
Given a graph $G=(V,E)$ and a parameter $\lambda > 0$, the problem is to sample an independent set $S$ with probability $\frac{\lambda^{|S|}}{Z_G(\lambda)}$ where $Z_G(\lambda) := \sum_{S \subset V: S {\rm~independent}} \lambda^{|S|}$ is the partition function.
An important work of Weitz~\cite{Weitz06} gave a deterministic fully polynomial time approximation scheme to estimate $Z_G(\lambda)$ for $\lambda$ up to the ``uniqueness threshold'', but the exponent of the runtime depends on the maximum degree $\Delta$ of $G$.
It is conjectured that the natural Markov chain for sampling independent sets mixes in polynomial time up to the uniqueness threshold.
Anari, Liu, and Oveis Gharan prove this conjecture and obtain a polynomial time algorithm to estimate $Z_G(\lambda)$ up to the uniqueness threshold for any graph (even with unbounded maximum degree).
They consider a pure $n$-dimensional simplicial complex for sampling independent sets,
and prove that $\gamma_j = \Theta(\frac{1}{n-j})$ for $0 \leq j \leq n-2$ by using the techniques from correlation decay.
Then it follows from \cref{t:main} that the Markov chain is fast mixing.
Note that it is crucial to have a different bound $\gamma_j$ for links of different dimension in \cref{t:main}, so even when $\gamma_{n-2} = \Theta(1)$ it is still possible to conclude fast mixing.

\subsubsection{Combinatorial Constructions of High Dimensional Expanders} \label{ss:LMY}

Recently, Liu, Mohanty, and Yang~\cite{LiuMY19} presented an interesting combinatorial construction of a sparse simplicial complex where all higher order random walks have a constant spectral gap.
Their construction is by taking a certain tensor product of a graph $G$ on $n$ vertices and a small $H$-dimensional complete simplicial complex ${\mathcal B}$ on $s$ vertices.

\begin{theorem}[\cite{LiuMY19}]\label{thm:lmy}
    Let $G$ be a $T$-regular triangle free graph on $n$ vertices. 
    There is an explicit family $(X^{(s,H, G)})_{H \ge 1, s \ge H+1}$ of simplicial complexes, satisfying the following properties:
    \begin{enumerate}
        \item $X^{(s, H, G)}$ is a pure $H$-dimensional simplicial complex
            with $\Theta(n)$ maximal faces.
        \item The spectral gap of the graphs of $j$ dimensional links of the complex $X^{(s,H, G)}$ satisfies
            \[ 1 - \gamma_j \ge \begin{cases}
                \frac{1}{2} & \textrm{ if } j \in [0, H-2],\\
                \parens*{\frac{1}{2} - \frac{1}{2(T2^H + 1)}}(1 - \sigma_2(G)) & \textrm{ if }j=
                -1,
            \end{cases}\]
            where $\sigma_2(G)$ is the second largest eigenvalue of the
            normalized adjacency matrix of $G$.
        \item 
            For any $-1 \leq j \leq H-2$,
            \[  \lambda_2(\DownW_{j+1}) = \lambda_2(\UpW_j) \le 1 -
            \Omega\parens*{\frac{1 - \sigma_2(G)}{T^2 \cdot j^2 \cdot (s - j)
            \cdot 2^j}},\]
    \end{enumerate}
\end{theorem}
The main technical part of their proof is in establishing Item (3) in \cref{thm:lmy}.
They use the special structures of their construction and the decomposition technique from~\cite{JerrumSTV04} to bound the spectral gap of the higher order random walks.
The authors ask the question whether the spectral property in Item (2) alone is enough to prove the fast mixing result in Item (3).
Note that Kaufman-Oppenheim's \cref{t:KO} does not apply in this regime.

Using~\cref{t:main}, we answer their question affirmatively, by deriving Item (3) from Item (2) in a black box fashion.
This slightly improves their bound and considerably simplifies their analysis.
\begin{corollary}
    Let $X := X^{(s, H, G)}$ be a complex from \cref{thm:lmy} satisfying Item (2). 
    For any $-1 \leq j \leq H-2$,
    \[ \lambda_2(\DownW_{j+1}) = \lambda_2(\UpW_j) \leq 1 -
    \Omega\parens*{ \frac{1- \sigma_2(G) }{ j \cdot 2^j} }.\]
\end{corollary}

\subsubsection{Longer Random Walks and Other Applications} \label{ss:long}

Consider the following generalization of the up-down walk where we take ``longer'' steps.
Initially, the random walk starts from an arbitrary $\alpha_1$ face of
dimension $a$ in $X$.
In each step $t \geq 1$, 
we sample a uniformly random face $\beta$ of dimension $b > a$ that contains $\alpha_t$, 
and set $\alpha_{t + 1}$ to be a uniformly random subset of $\beta$ of dimension $a$. 
We call this the $a$-th up-down walk through the $b$-th dimension, and denote its transition matrix by $\UpW_{a, b}$.
The $k$-th up-down walk defined before is the special case $\UpW_{k,k+1}$.
Dinur and Kaufman~\cite{DinurK17} derived the following result about $\UpW_{a,b}$ from the result about the ordinary up-down walks.

\begin{corollary}[\cite{DinurK17}]\label{cor:m1m2dk}
    Let $X$ be a $d$-dimensional pure simplicial complex. 
    If $X$ is a $\gamma$-local-spectral expander, 
    then for any $0 \leq a < b \leq d-1$,
    \[ \lambda_2(\UpW_{a, b}) \le \frac{a + 1}{b + 1} + O(a 
    (b - a) \gamma).\]
\end{corollary}

Using \cref{t:main}, we obtain the following improved bound.
See \cref{sec:eigenbd} for the proof.

\begin{corollary}\label{c:fakem1m2}
    Let $X$ be a $d$-dimensional pure simplicial complex. 
    If $X$ is a $\gamma$-local-spectral expander, 
    then for any $0 \leq a < b \leq d-1$,
    \[ \lambda_2(\UpW_{a, b}) \le (1+\gamma)^{b-a} \cdot \frac{a+1}{b+1}.\]
    In particular, if $\gamma \le \frac{\ee}{b - a}$ for some $0 \le \ee \le 1$, then $\lambda_2(\UpW_{a, b}) \le e^\ee \cdot \frac{a+1}{b+1}.$
\end{corollary}

Whereas the bound from \cref{cor:m1m2dk} requires $\gamma = O(\frac{1}{b \cdot (b-a)})$ to give a nontrivial upper bound on the second eigenvalue of $\UpW_{a,b}$, 
\cref{c:fakem1m2} only requires $\gamma \le O(\frac{1}{b-a})$ to give a comparable bound.

\cref{cor:m1m2dk} has found applications in agreement testing and coding theory~\cite{DinurK17,DinurHKNT19, AlevJQST19}. 
We believe that \cref{c:fakem1m2} can be of independent interest because of
those applications. 
One potential application would be in constructing double
samplers from Ramanujan complexes under a weaker expansion assumption~\cite{DinurK17}.

\subsection{Related Work}

\subsubsection*{Higher Order Random Walks and Applications}

Our work follows a sequence of works~\cite{KaufmanM17,DinurK17,Oppenheim18,KaufmanO18,DiksteinDFH18} which use the spectral properties of the links of simplical complexes to analyze higher order random walks.
Higher order random walks on simplicial complexes were first introduced by Kaufman and Mass~\cite{KaufmanM17}.
They formulated related but more combinatorial notions of skeleton expansion and colorful expansion to establish fast mixing of higher order random walks.
Dinur and Kaufman~\cite{DinurK17} introduced the definition of two-sided $\gamma$-local-spectral expanders, which is similar to \cref{d:local} but requires all but the first eigenvalue to have absolute value at most $\gamma$ (i.e.~it also controls the negative eigenvalues).
They used this stronger assumption to prove a similar theorem as in \cref{t:KO},
and applied it to construct efficient agreement tester with applications to PCP constructions.
The one-sided $\gamma$-local-expander in \cref{d:local} was first studied by Oppenheim~\cite{Oppenheim18}, where he proved \cref{t:opp}.
Then, Kaufman and Oppenheim~\cite{KaufmanO18} strengthened the result in~\cite{DinurK17} and prove \cref{t:KO}.

Dikstein, Dinur, Filmus and Harsha~\cite{DiksteinDFH18} studied an alternative definition of high dimensional expanders, 
based on the operator norm of the difference between the (non-lazy) up-down and down-up operators. 
Using this definition, they show that it is possible to approximately characterize all the eigenvalues and eigenvectors of higher order random walks. 
Their techniques were used in \cite{AlevJT19} to analyze the ``swap walks'' on high dimensional expanders, with applications in designing good approximation algorithms for solving constraint satisfaction problems on high-dimensional expanders. 
Independently, the same ``swap walks'' were also studied by \cite{DiksteinD19} under the name ``complement walks'', where applications in agreement testing were given.

The results in higher order random walks have also found applications in coding theory.
The double samplers in~\cite{DinurK17} are used in~\cite{DinurHKNT19} to design an efficient algorithm to decode direct product codes over high dimensional expanders. 
The swap walks in~\cite{AlevJT19} are used in~\cite{AlevJQST19} to recover the same result and also to design an efficient algorithm to decode direct sum codes over high dimensional expanders.

\subsubsection*{Analyzing Mixing Times of Markov Chains}

Mixing time of Markov chains is an extensively studied topic with various applications (see e.g.~\cite{WilmerLP09, MontenegroT05}).
There are several well-developed approaches to bound the mixing time of a Markov chain.
Perhaps the most widely used approach is the coupling method (e.g.~\cite{Aldous83,BubleyD97}), which has applications in sampling graph colorings (e.g.~\cite{Jerrum95,Vigoda00}) and many other problems (see~\cite{WilmerLP09}).
The canonical path (or more generally multicommodity flow) method developed in~\cite{JerrumS89, Sinclair92, Sinclair93} was used in the important problem of sampling perfect matchings in bipartite graphs~\cite{JerrumS89, JerrumSV04} and other problems including sampling matroid bases~\cite{FederM92}.
Geometric methods are used in the important problem of sampling a random point
in a convex body~\cite{DyerFK91, LovaszV06}.
Analytical methods such as (modified) log-Sobolev inequalities and Nash inequalities~\cite{DiaconisSC96, BobkovT06} are useful in proving sharp bounds on mixing time, e.g.~a recent paper~\cite{CryanGM19} used a modified log-Sobolev inequality to prove optimal mixing time of the natural Markov chain on sampling matroid bases.

The simplicial complex approach studied in this paper is quite different from the above approaches.
It is linear algebraic and designed to bound the second eigenvalue directly using ideas from simplicial complexes.
On the other hand, the coupling method is probabilistic and designed to compare two random processes, while the canonical path method and the geometric method are designed to bound the underlying expansion of the graph or the geometric object.
The analytical methods are more diffcult to apply and are not as widely applicable, but when they work they could be used to prove very sharp results.

\section{Preliminaries}\label{sec:prelim}

\subsection{Linear Algebra}

\subsubsection*{Vectors and Inner-Products}

Bold faces will be used for scalar functions, i.e.~$\eff \in \RR^V$.  
The notation $\one_V \in \RR^V$ will be reserved for the all-one vector in $\RR^V$;
the subscript may be omitted when the vector space $\RR^V$ is clear from the context.

Throughout this text, we use $\Pi \in \RR^V$ to denote various probability distributions,
i.e.~$\sum_{x \in V} \Pi(x) = 1$ and $\Pi(x) \geq 0$ for $x \in V$. 
Given $\eff, \gee \in \RR^V$, we use the notations $\langle \eff, \gee\rangle_{\Pi}$ and $\norm{\eff}_\Pi$ to denote the inner-product and the norm with respect to the distribution $\Pi$, i.e.
\[ \langle \eff, \gee\rangle_{\Pi} = \sum_x \Pi(x)\eff(x)\gee(x)  ~~\textrm{ and
}~~ \norm{\eff}_\Pi^2 = \langle \eff, \eff\rangle_{\Pi}.\]
We reserve $\langle \eff, \gee \rangle = \sum_x \eff(x) \gee(x)$ for the standard inner-product.
Given $\eff \in \RR^V$, we write $\norm{\eff}_{\ell_1} = \sum_{x \in V} |\eff(x)|$ for its $\ell_1$-norm, and $\norm{\eff}_{\ell_2} = (\sum_{x \in V} \eff(x)^2)^{\frac12}$ for its $\ell_2$-norm.

\subsubsection*{Matrices and Eigenvalues}

Serif faces will be used for matrices, i.e.~$\Aye \in \RR^{V \times V}$.
Let $G=(V,E)$ be an edge-weighted undirected graph with a weight $w_e > 0$ on each edge $e \in E$.
The adjacency matrix of $G$ is denoted by $\Aye_G \in \RR^{V \times V}$ with $\Aye_G(u, v) = w_{uv}$ for $uv \in E$ and $\Aye_G(u,v) = 0$ for $uv \notin E$.
The diagonal degree matrix of $G$ is denoted by $\Dee_G$ with $\Dee_G(v,v) = \deg(v) = \sum_{u: uv \in E} w_{uv}$ for $v \in V$.
The random walk matrix of $G$ is denoted by $\Emm_G := \Dee_G^{-1} \Aye_G$.
Note that $\Emm_G$ is a row-stochastic matrix where every row sums to one.
Throughout this text, we will use $\Emm \in \RR^{U \times V}$ to denote row-stochastic operators, where $\Emm \one_V = \one_U$.

The adjoint of the operator $\Bee \in \RR^{V \times U}$, with respect to the
inner-products defined by $\Pi_U$ and $\Pi_V$ on $U$ and $V$, is the unique operator $\Bee^* \in \RR^{U \times V}$ such that
\[ \langle \eff, \Bee \gee\rangle_{\Pi_U} = \langle \Bee^* \eff,
\gee\rangle_{\Pi_V}~~\textrm{ for all } \eff \in \RR^U, \gee \in \RR^V.\]
If $U = V$ and $\Pi_U = \Pi_V$, the
operator $\Bee$ is called self-adjoint if $\Bee^*= \Bee$. 
Note that a real symmetric matrix is self-adjoint with respect to the standard inner-product. 

If $\Emm$ is a row-stochastic self-adjoint operator (with respect to the stationary distribution $\Pi$), 
then the Markov chain described by $\Emm$ is called reversible.
The random walk operator of an edge-weighted undirected graph is described by the self-adjoint row-stochastic operator $\Emm_G$ (with respect to the stationary distribution $\Pi = \Dee_G \one / \sum_v \deg(v)$) and is a reversible Markov chain.

Let $\Why \in \RR^{V \times V}$ be a self-adjoint operator with respect to the inner-product defined by $\Pi$.
It is a fundamental result in linear algebra that $\Why$ has only real eigenvalues, and an orthonormal set of eigenvectors with respect to the inner-product defined by $\Pi$, i.e.~$\langle \eff, \gee \rangle_\Pi = 0$ for eigenvectors $\eff \neq \gee$. 
We write $\lambda_i(\Why)$ for the $i$-th largest eigenvalue of $\Why$ so that $\lambda_1(\Why) \ge \ldots \ge \lambda_{|V|}(\Why)$,
and write $\lambda_{\min}(\Why)$ for the smallest eigenvalue of $\Why$, i.e.~$\lambda_{\min}(\Why) = \lambda_{|V|}(\Why)$.
The largest eigenvalue
$\lambda_1(\Why)$ of a self-adjoint matrix $\Why$ with respect to the measure
$\Pi$ obeys the variational formula
\begin{equation}
    \lambda_1(\Why) = \max\set*{\langle \eff, \Why \eff\rangle_{\Pi} : \eff \in
    \RR^V, \norm{ \eff }_{\Pi} = 1}.\label{eq:var-form}\tag{variational formula}
\end{equation}
It is well-known that the maximizers of the \ref{eq:var-form} are precisely the
unit eigenvectors of $\Why$ corresponding to $\lambda_1(\Why)$, i.e.~$\Why \eff =
\lambda_1(\Why) \cdot \eff$ if and only if $\eff$ maximizes the RHS in the
\ref{eq:var-form}.

Given an arbitrary operator $\Bee \in \RR^{V
\times U}$ we will write $\sigma_i(\Bee)$ for the $i$-th largest singular value of $\Bee$ so that $\sigma_1(\Bee) \ge \ldots \ge \sigma_{\min\set*{|U|,|V|}}(\Bee)$. 
It is well known that the singular values of a real operator $\Bee$ coincide with the eigenvalues of the self-adjoint operator $\Bee\Bee^*$.

A self-adjoint operator $\Aye \in \RR^{V \times V}$ with respect to inner-product defined by $\Pi$ is called positive semi-definite, denoted by $\Aye \succeq_{\Pi} 0$, if it satisfies $\langle \eff, \Aye \eff\rangle_{\Pi} \ge 0$ for all $\eff \in \RR^V$.
The condition is equivalent to the condition that $\lambda_{\min}(\Aye) \ge 0$. 
For self-adjoint operators $\Aye \in \RR^{V \times V}$ and $\Bee \in \RR^{V
\times V}$ with respect to the same inner-product defined by $\Pi$, we
will write $\Aye \preceq_\Pi \Bee$ if
\[ \langle \eff, \Aye \eff\rangle_\Pi \le \langle \eff, \Bee \eff\rangle_{\Pi}
~~\textrm{ for all }~\eff \in \RR^V.\]
This is equivalent to $\Aye -\Bee$ being positive-semidefinite, i.e.~$\Aye - \Bee \succeq_{\Pi} 0$.
If $\Pi$ is just the standard inner-product, we will drop the subscript $\Pi$.

We will use the following results about eigenvalues in \cref{sec:eigenbd} and \cref{sec:sampling}; see e.g.~\cite{Bhatia2013}.

\begin{fact}\label{fac:simple-la}
    Let $\Aye \in \RR^{U \times V}$ and $\Bee \in \RR^{V \times U}$. Then, the non-zero spectrum of $\Aye\Bee$ coincides with that of $\Bee\Aye$ with the same multiplicity.
\end{fact}

\begin{fact}\label{fac:simple-loew}
    Let $\Aye, \Bee \in \RR^{V \times V}$ be two self-adjoint matrices with respect to the inner-product defined by $\Pi$ satisfying
    $\Aye \preceq_\Pi \Bee$.
    Then, $\lambda_i(\Aye) \le \lambda_i(\Bee)$ for all $1 \leq i \leq |V|$.
\end{fact}

\begin{theorem}[Cauchy Interlacing Theorem]\label{thm:cauchy-int}
    Let $\Aye \in \RR^{V \times V}$ be a symmetric matrix and $\Bee \in
    \RR^{U \times U}$ be a principal submatrix of $\Aye$.
    Let $n = |V|$ and $m = |U|$. For any $0 \leq j \leq m$,
    \[ \lambda_j(\Aye) \ge \lambda_j(\Bee) \ge \lambda_{n - m + j}(\Aye). \]
\end{theorem}

\begin{theorem}[Weyl Interlacing Theorem]\label{thm:weyl-int}
    Let $\Aye, \Bee \in \RR^{V \times V}$ be two symmetric matrices.
    For any $i,j$,
    \[\lambda_{i+ j - 1}(\Aye + \Bee) \le \lambda_i(\Aye) + \lambda_j(\Bee).\]
\end{theorem}

\subsection{Simplicial Complexes}

A simplicial complex $X$ is a collection of subsets that is downward closed, i.e.~if $\beta \in X$ and $\alpha \subset \beta$ then $\alpha \in X$.
The elements $\alpha, \beta$ in $X$ are called faces/simplices of $X$.
The dimension of a face $\alpha$ is defined as $|\alpha|-1$, e.g.~an edge is of dimension $1$, a vertex/singleton is of dimension $0$, the empty set is of dimension $-1$.
The collection of faces of dimension $j$ is denoted by $X(j)$.
The dimension of a simplicial complex is defined as the maximum dimension of its faces.
A $d$-dimensional simplicial complex is called pure if every maximal face is of dimension $d$.
All simplicial complexes considered in this paper are pure.

\subsubsection*{Weighted Simplicial Complexes}

A simplicial complex $X$ can be equipped with a weighted function which assigns a positive weight to each face of $X$.
We follow the formalism of~\cite{DiksteinDFH18} where the weight function is a probability distribution $\Pi$ on the faces of the same dimension.
Let $X$ be a $d$-dimensional simplicial complex.
Given a probability distribution $\Pi := \Pi_d$ on $X(d)$, 
we can inductively obtain probability distributions $\Pi_j$ on all $X(j)$ by considering the marginal distributions, i.e.~
\begin{equation}
    \Pi_{j}(\alpha) = \frac{1}{j + 2}  \sum_{\beta \in X(j+ 1), \atop \beta \supset \alpha} \Pi_{j + 1}(\beta). \label{eq:onestep}
\end{equation}
Equivalently, we can understand $\Pi_j$ as the probability distribution of the following random process:
Sample a random face $\beta \in X(d)$ using the probability distribution $\Pi_d$, and then sample a uniform random subset of $\beta$ in $X(j)$.
The pair $(X, \Pi)$ will be referred as a weighted simplicial complex. 
We write $(X, \Pi)$ simply as $X$ when $\Pi$ is the uniform distribution.

\subsubsection*{Links and Graphs} 

Let $(X,\Pi)$ be a pure $d$-dimensional weighted simplicial complex.
The link $X_\alpha$ of a face $\alpha$ is the simplicial complex defined as
\[
    X_{\alpha} := \{ \beta \setminus \alpha \mid \beta \in X, \beta \supset \alpha \}.
\]
The probability distributions $\Pi_0, \ldots, \Pi_d$ on $X$ can naturally be used to define the probability distributions $\Pi^{\alpha}_0, \ldots, \Pi^{\alpha}_{d-|\alpha|}$ on $X_{\alpha}$ using conditional probability.
Suppose $\alpha \in X(j)$.
The probability distribution $\Pi_l^{\alpha}$ for $X_{\alpha}(l)$ is defined as
\begin{equation}
    \Pi^\alpha_l(\tau) 
    = \Pr_{\beta \sim \Pi_{j + 1 +l}}\sqbr*{ \beta = \alpha \cup \tau \mid \beta \supset \alpha } 
    = \frac{\Pi_{j + l + 1}(\alpha \cup \tau)}{\binom{|\alpha \cup \tau|}{|\alpha|} \cdot \Pi_j(\alpha) }
    \qquad {\rm~for~} \tau \in X_{\alpha}(l),
    \label{eq:link-def}
\end{equation}
where the latter equality is obtained by applying \cref{eq:onestep} repeatedly. 

Given a link $X_{\alpha}$, the graph $G_{\alpha} = (X_{\alpha}(0), X_{\alpha}(1), \Pi^{\alpha}_1)$ is defined as the $1$-skeleton of $X_{\alpha}$.
More explicitly, each singleton $\{v\}$ in $X_{\alpha}$ is a vertex $v$ in $G_{\alpha}$, each pair $\{u,v\}$ in $X_{\alpha}$ is an edge $uv$ in $G_{\alpha}$, and the weight of $uv$ in $G_{\alpha}$ is equal to $\Pi_1^{\alpha}(\{u,v\})$. 
A simple observation is that if $X$ is a pure $d$-dimensional simplicial complex and $\Pi$ is the uniform distribution on $X(d)$, then for any $\alpha \in X(d-2)$ the weighting $\Pi_1^{\alpha}$ on the edges of $G_{\alpha}$ is uniform.
We will use this observation in \cref{sec:sampling}.

\subsection{Local Spectral Expanders} \label{ss:local-spectral}

\subsubsection*{Random Walk Matrices}

The definition of local spectral expanders will be based on the random walk matrix of $G_{\alpha}$.
Let $\Aye_{\alpha}$ be the adjacency matrix of $G_{\alpha}$.
Let $\Dee_{\alpha}$ be the diagonal degree matrix where $\Dee_{\alpha}(x,x) = \sum_{y} \Aye_{\alpha}(x,y) = 2\Pi_0^{\alpha}(x)$ where the last equality is by \cref{eq:onestep}.
The random walk matrix $\Emm_{\alpha}$ of $G_{\alpha}$ is defined as $\Emm_{\alpha} := \Dee_{\alpha}^{-1} \Aye_{\alpha}$, with
\[ \Emm_\alpha(x, y) = \frac{\Pi_1^{\alpha}(x, y)}{2 \Pi^\alpha_0(x)} \quad \textrm{ for all } \set*{x, y} \in X_\alpha(1).\]
The distribution $\Pi_0^{\alpha}$ is the stationary distribution of $\Emm_{\alpha}$, as 
\[(\Pi_0^{\alpha})^\top \Emm_{\alpha} = (\Pi_0^{\alpha})^\top \Dee_{\alpha}^{-1} \Aye_{\alpha} = \one^\top \Aye_{\alpha} = (\Pi_0^{\alpha})^\top.\]
The matrix $\Emm_{\alpha}$ is self-adjoint with respect to the inner-product defined by $\Pi_0^{\alpha}$, as 
\[\langle \eff, \Emm_{\alpha} \gee \rangle_{\Pi_0^{\alpha}}
= \langle \eff, \Dee_{\alpha}^{-1} \Aye_{\alpha} \gee \rangle_{\Pi_0^{\alpha}}
= \langle \eff, \Aye_{\alpha} \gee \rangle
= \langle \Aye_{\alpha} \eff, \gee \rangle
= \langle \Dee_{\alpha}^{-1} \Aye_{\alpha} \eff, \gee \rangle_{\Pi_0^{\alpha}}
= \langle \Emm_{\alpha} \eff, \gee \rangle_{\Pi_0^{\alpha}}.
\]
So, $\Emm_{\alpha}$ have only real eigenvalues, and an orthonormal basis of eigenvectors with respect to the inner-product defined by $\Pi_0^{\alpha}$.
The largest eigenvalue of $\Emm_{\alpha}$ is $1$, as $\Emm_{\alpha} \one = \one$ and $\Emm_{\alpha}$ is row-stochastic.

Given a vector $\eff$, we will be interested in writing it as $\eff = \eff^{\one} + \eff^{\perp\one}$, so that $\eff^{\one} = c \one$ for some scalar $c$ and $\langle \eff^{\one}, \eff^{\perp\one} \rangle_{\Pi_0^{\alpha}} = 0$.
It follows that $c = \frac{\langle \eff, \one \rangle_{\Pi_0^{\alpha}}}{ \langle \one, \one \rangle_{\Pi_0^{\alpha}}} = \langle \eff, \one \rangle_{\Pi_0^{\alpha}} = \Exp_{x \sim \Pi_0^\alpha}[\eff(x)]$.
We write $\Jay_\alpha = \one (\Pi_0^\alpha)^\top$ as the operator to map $\eff$ to $\eff^{\one}$, so that
\begin{equation} 
    \Jay_\alpha \eff 
    = (\one (\Pi_0^\alpha)^\top) \eff
    = \langle \eff, \Pi_0^{\alpha} \rangle  \cdot \one 
    = \Exp_{x \sim \Pi_0^\alpha}[\eff(x)] \cdot \one 
    = \eff^{\one} \label{eq:J}\tag{projector to constant functions}.
\end{equation}

\subsubsection*{Local Spectral Expanders and Oppenheim's Theorem}

Let $(X, \Pi)$ be a pure $d$-dimensional weighted simplicial complex. 
Define  
\[\gamma_j := \gamma_j(X, \Pi) = \max_{\alpha \in X(j)}
\lambda_2(\Emm_\alpha) ~~\textrm{
    for all $j =-1,\ldots,d-2$},\] 
where $\lambda_2(\Emm_\alpha)$ is the
second largest eigenvalue of the operator $\Emm_\alpha$.
We say $X$ is a $\gamma$-local-spectral expander if $\gamma_i \le \gamma$ for $-1 \leq i \leq d-2$.

Oppenheim's Theorem relates the second eigenvalue of the graph of a lower-dimensional link to that of a higher-dimensional link.
It works for any weighted simplicial complex with a ``balanced'' weight function $w$, where for any $\alpha \in X(k)$ and any $-1 \leq k \leq d-1$ it holds that
\[w(\alpha) = c_k \sum_{\beta \in X(k+1), \atop \beta \supset \alpha} w(\beta)\]
for some constant $c_k$ that only depends on $k$.
Note that the weight function in \cref{eq:onestep} satisfies this condition with $c_k = 1/(k+2)$.

\begin{theorem}[Oppenheim's Theorem]\label{thm:oppenheim}
    Let $(X,\Pi)$ be a pure $d$-dimenisonal weighted simplicial complex where $\Pi$ satisfies \cref{eq:onestep}. 
    For any $0 \leq j \leq d-2$, 
    if $G_\alpha$ is connected for every $\alpha \in X(j-1)$, then
    \[ \gamma_{j - 1}\le \frac{\gamma_j}{1 - \gamma_j}.\]
\end{theorem}

An inductive argument proves the following corollary.

\begin{corollary}[Oppenheim's Corollary]\label{cor:oppenheim}
    Let $(X,\Pi)$ be a pure $d$-dimenisonal weighted simplicial complex where $\Pi$ satisfies \cref{eq:onestep}.
    If $G_\alpha$ is connected for every $\alpha \in X(k)$ and every $k \leq d-2$, 
    then
    \[ \gamma_j \le \frac{\gamma_{d - 2}}{1 - (d-2- j)  \cdot \gamma_{d-2}}.\]
\end{corollary}

\subsection{Higher Order Random Walks}\label{sec:updown-defn}

\subsubsection*{Up and Down Operators}

Let $(X, \Pi)$ be a pure $d$-dimensional weighted simplicial complex.
In the following definitions, 
$\alpha \in X(k)$, $\beta \in X(k+1)$, $\eff \in \RR^{X(j)}$, $\gee \in \RR^{X(k+1)}$, and $j \in \{-1,0,\ldots,d-1\}$.

The $j$-th up operator $\Up_j: \RR^{X(j)} \to \RR^{X(j + 1)}$ is defined as 
\begin{equation}
    [\Up_j \eff](\beta) 
    = \frac{1}{j+2} \sum_{x \in \beta} \eff(\beta \backslash x) 
    = \sum_{\alpha \subset \beta,\atop \alpha \in X(j)} \frac{\eff(\alpha)}{j+2}.
    \label{eq:up-def}\tag{up operator}
\end{equation}
The $(j+1)$-st down operator $\Dee_{j+1}: \RR^{X(j+1)} \to \RR^{X(j)}$ is defined as
\[ [\Dee_{j+1} \gee](\alpha) 
= \sum_{x \in X_\alpha(0)} \frac{\Pi_{j+1}(\alpha \cup x)}{(j + 2) \Pi_j(\alpha)} \cdot \gee(\alpha \cup x)
= \sum_{\beta \supset \alpha,
\atop \beta \in X(j+1)} \frac{\Pi_{j+1}(\beta) \cdot \gee(\beta)}{(j + 2) \Pi_j(\alpha)}  
\label{eq:down-def}\tag{down operator}  \]
It can be checked~\cite{KaufmanO18,DiksteinDFH18} that the adjoint of $\Up_j: \RR^{X(j)} \to \RR^{X(j + 1)}$ with respect to the inner-products defined by $\Pi_{j+1} \in \RR^{X(j + 1)}$ and $\Pi_j \in \RR^{X(j)}$ is $\Dee_{j+1}$, i.e.
\begin{equation} \label{eq:adjoint}
    \langle \gee, \Up_j \eff\rangle_{\Pi_{j+1}} 
    = \langle \Dee_{j+1} \gee, \eff \rangle_{\Pi_{j}}  
    ~~\textrm{ for all } \gee \in \RR^{X(j+1)}, \eff \in \RR^{X(j)}.
    \tag{adjointness}
\end{equation}
And it follows that the adjoint of $\Dee_{j+1}$ with respect to the inner-products defined by $\Pi_j$ and $\Pi_{j+1}$ is $\Up_j$, i.e.
    $\langle \eff, \Dee_{j+1} \gee\rangle_{\Pi_{j}} 
    = \langle \Up_{j} \eff, \gee \rangle_{\Pi_{j+1}}
    $  
    for all $\gee \in \RR^{X(j+1)}, \eff \in \RR^{X(j)}$.

\begin{remark}
    We have stayed consistent with the notations in~\cite{DiksteinDFH18}, and named $\Up_j$ and $\Dee_{j+1}$ up and down operators
    with their right-action on functions (or vectors) in mind. However, in terms of random walks, $\Up_j$ describes a random down-movement from $X(j+1)$ to $X(j)$, whereas $\Dee_{j+1}$ describes a random up-movement from $X(j)$ to $X(j+1)$, since the action of the probability
    distribution is from the left.
\end{remark}

\subsubsection*{Down-Up Walk, Up-Down Walk, and Non-Lazy Up-Down Walk}

We use the up and down operators to define three random walk operators on $X(j)$.
The $j$-th down-up walk $\DownW_j$ and the $j$-th up-down walk $\UpW_j$ are defined as
\begin{equation}
    \DownW_j = \Up_{j-1} \Dee_j
    \quad {\rm and} \quad
    \UpW_j = \Dee_{j+1}\Up_j.
    \tag{down-up walk, up-down walk}
\end{equation}
As $\Up_{i}^*= \Dee_{i + 1}$, it is easy to observe that these operators are
positive semi-definite. One useful property of $\UpW_j$ and $\DownW_j$ is that they have the same non-zero spectrum with the same multiplicity by \cref{fac:simple-la},
and in particular $\lambda_2(\UpW_j) = \lambda_2(\DownW_j)$.
Also, we define the $j$-th non-lazy up-down walk as
\begin{equation}
    \NUpW_j = \frac{j+2}{j+1} \parens*{\UpW_j - \frac{1}{j+2} \Ide },\label{eq:nupw-def}\tag{non-lazy up-down walk}
\end{equation}
which is the up-down walk conditioned on not looping. 
It follows from the \ref{eq:adjoint} of $\Up_j$ and $\Dee_{j+1}$ that all $\DownW_{j}$, $\UpW_{j}$, and $\NUpW_j$ are self-adjoint with respect to the inner-product defined by $\Pi_j$, e.g. given any $\eff_1, \eff_2 \in \RR^{X(j)}$,
\begin{equation}
    \langle \eff_1, \UpW_{j} \eff_2 \rangle_{\Pi_j}
    = \langle \eff_1, \Dee_{j+1} \Up_j \eff_2 \rangle_{\Pi_j}
    = \langle \Up_{j} \eff_1,  \Up_j \eff_2 \rangle_{\Pi_{j+1}}
    = \langle \Dee_{j+1} \Up_{j} \eff_1,  \eff_2 \rangle_{\Pi_{j}}
    = \langle \UpW_{j} \eff_1, \eff_2 \rangle_{\Pi_j}.
\end{equation}
These imply that $\Pi_j$ is the stationary distribution for all these random walks $\DownW_{j}$, $\UpW_{j}$, and $\NUpW_j$,
e.g. putting $\eff_1 = \one$ and $\eff_2 = \chi_i$ into $\Pi_j^\top$, then
\begin{equation} \label{eq:stationary}
    \Pi_j^\top \UpW_j(i)
    = \langle \one, \UpW_j \chi_i \rangle_{\Pi_j}
    = \langle \UpW_j \one,  \chi_i \rangle_{\Pi_j}
    = \langle \one,  \chi_i \rangle_{\Pi_j}
    = \Pi_j^\top(i)
    \quad \implies \quad
    \Pi_j^\top \UpW_j = \Pi_j^\top.
\end{equation}

{\bf Combinatorial Interpretation:}
We can understand the higher order random walks as a random walk on a bipartite graph between $X(j)$ and $X(j+1)$ as explained in~\cite{AnariLOV18, DinurK17}.
Consider the bipartite graph $H = (X(j), X(j + 1), E)$ in which a face $\alpha
\in X(j)$ and a face $\beta \in X(j + 1)$ are connected if and only if $\alpha \subset \beta$. 
The edge $\set{\alpha, \beta} \in H$ is assigned the weight
$\frac{1}{j+ 2} \cdot \Pi_{j+ 1}(\beta)$. Using \cref{eq:onestep}, it can be seen that the weighted degree of any $\alpha \in X(j)$ is $\Pi_j(\alpha)$. 
And the weighted degree of any $\beta \in X(j+ 1)$ is exactly $\Pi_{j+1}(\beta)$.  
Thus, the graph $H$ has the (weighted) random walk matrix
\[ \Emm_H = \begin{pmatrix} 0 & \Up_j\\
\Dee_{j + 1} & 0\end{pmatrix}.\]
One step of the down-up walk $\DownW_{j + 1}$ can be thought as a two step random walk in $\Emm_H$:
starting from some $\beta \in X(j+1)$, 
the random walk will go down from $\beta \in X(j+ 1)$ to $\alpha \in X(j)$ by dropping an element of $\beta$, which is chosen uniformly at random as prescribed by $\Up_{j}$,
and then the random walk will go up from $\alpha \in X(j)$ to a random face $\beta' \in X(j + 1)$ which contains $\alpha$ as prescribed by $\Dee_{j+1}$. 
Similarly, one step of the up-down walk $\UpW_{j}$ can be thought as a two step random walk in $\Emm_H$ starting from some $\alpha \in X(j)$.
More precisely,
\[ \Emm_H^2 = \begin{pmatrix}
    \Up_j \Dee_{j + 1} & 0 \\
    0 & \Dee_{j + 1}\Up_j & 0
\end{pmatrix} = \begin{pmatrix} \DownW_{j+1} & 0\\
0 & \UpW_{j} \end{pmatrix}.\] 
It is instructive to check that when the distribution $\Pi$ of the simplicial complex is the uniform distribution, then the down-up walks and the up-down walks are as described as in the introduction.

\subsubsection*{Longer Random Walks}

Suppose now $-1 \le a < b \le d$. We define the up-down walk on $X(a)$ through
$X(b)$ to be
\[ \UpW_{a, b} = \Dee_{a+1} \cdots \Dee_{b} \cdot \Up_{b -1} \cdots
\Up_{a}.\]
Similar to the intuition that was presented about the up-down
and the down-up walks, we can think of $\UpW_{a, b}$ as simulating
two-steps of the
random walk starting from some face $\alpha \in X(a)$ 
on the weighted bipartite graph $H = (X(a), X(b), E)$
where $\set*{\alpha, \beta}$ is an edge of this graph with weight proportional to $\Pi_b(\beta)$ whenever $\alpha \subset \beta$.

\subsection{Mixing Times of Markov Chains}

Recall that two distributions $\Pi$ and $\Pi'$ are said to be
\ref{eq:eps-close} if
\begin{equation}
    \norm{\Pi - \Pi'}_{\ell_1} = \sum_{x \in V} |\Pi(x) - \Pi'(x)| \le \ee. \tag{$\ee$-close}
    \label{eq:eps-close}
\end{equation}
The \ref{eq:mix} $T(\ee, \Pii)$ of the random walk operator $\Pii$ is defined 
to be the least time step where the distribution of the random walk is $\ee$-close to the stationary distribution $\Pi$ of $\Pii$ in the $\ell_1$ distance, i.e.~
\begin{equation}
    T(\ee, \Pii) = \min\set*{ t \in \NN_{\ge 0} : \norm{\Pii^t(x, \bullet) -
    \Pi}_{\ell_1}\le \ee \textrm{ for all }x \in V}. \label{eq:mix}\tag{mixing time}
\end{equation}

For our applications in sampling in \cref{sec:sampling}, we will use the following well known relation between the mixing time of the random walk and the spectral gap of its transition matrix (see e.g.~\cite[Proposition 1.12]{MontenegroT05}). 

\begin{theorem}[Spectral Mixing Time Bound]\label{thm:spec-mix-bd}
    Let $\Pii \in \RR^{V \times V}$ be a random walk matrix with stationary distribution $\Pi$. One has,
    \[ T(\ee, \Pii) \le\frac{1}{1 - \sigma_2(\Pii)} \cdot \log \frac{1}{\ee
    \cdot \min_{x \in V} \Pi(x)},\]
    where $\sigma_2(\Pii)$ is the second largest singular value of $\Pii$.
\end{theorem}

The operator of importance for us will be $\Pii = \DownW_j$. 
As this operator is positive semi-definite as explained in \cref{sec:updown-defn}, 
we have $\sigma_2(\DownW_j) = \lambda_2(\DownW_j)$. 
Also, recall from \cref{sec:updown-defn} that the stationary distribution of $\UpW_j$ is $\Pi_j$, we obtain
\[ T(\ee, \DownW_j) \le \frac{1}{1 - \lambda_2(\DownW_j)} \cdot \log
\frac{1}{\ee \cdot \min_{\alpha \in X(j)} \Pi_j(\alpha)}.\]

\subsubsection*{Approximate Sampling and Approximate Counting} \label{ss:countsampl}

There is a well-known equivalence between approximate
sampling and approximate counting for self-reducible problems. 
Let $\Omega := \set*{\Omega_s}_s$ be a collection of sets parametrized by some
strings $s$, e.g.~$s$ can be describing a graph and $\Omega_s$ the set of perfect matchings in $G$.
Suppose a randomized algorithm $\Aa$ is given whose output distribution is
described by $\mu_{\Aa(s)}$. 
Then $\Aa$ is called a fully polynomial time randomized approximate uniform
sampler (FPRAUS) for $\Omega_s$, if for every input string $s$
we have
\[ \norm{\mu_{\Aa(s)} - \Pi_{\Omega_s} }_{\ell_1} \leq \delta,\]
where $\Pi_{\Omega_s}$ describes the uniform distribution over $\Omega_s$ 
and the algorithm $\Aa$ runs in time $\poly(\langle s \rangle,
\log(1/\delta))$, where $\langle s \rangle$ denotes the size of the input.

Similarly, an algorithm $\Aa'$ is called a fully polynomial time
randomized approximation scheme (FPRAS) for $\Omega$, if we for every input $s$
we have
\[ \Pr[ (1- \delta) \cdot |\Omega_s| \le \Aa'(s) \le (1 + \delta) \cdot |\Omega_s| ] \ge 1- \ee,\]
and the algorithm $\Aa'$ runs in time $\poly(\langle s \rangle, 1/\ee,
\log(1/\delta))$.

A well-known result proven in \cite{JerrumVV86} asserts that approximate
counting and approximate sampling are equivalent for self-reducible problems.
\begin{theorem}[Informal]
    For self-reducible sets $\Omega$ in {\sc NP},
    the existence of an FPRAS for $\Omega$ is equivalent to the
    existence of an FPRAUS for $\Omega$.
\end{theorem}
In \cref{sec:sampling}, we will give approximate samplers for independent sets a graph, and it follows from this equivalence we can also approximately count the number of independent sets in the graph.

\section{Eigenvalue Bounds for Higher Order Random Walks}\label{sec:eigenbd}

Our main result is a quantitative generalization of the basic fact that a pure $d$-dimensional simplicial complex $X$ is gallery connected (i.e.~$\lambda_2(\DownW_d) < 1$) if and only if the graph $G_\alpha$ is connected for every $\alpha \in X$ up to dimension $d-2$ (i.e.~$\gamma_j < 1$ for $-1 \leq j \leq d-2$).
The statement is essentially the same as in \cref{t:main} but for more general weighted simplicial complexes.

\begin{restatable}{theorem}{main}\label{thm:main}
    Let $(X,\Pi)$ be a pure $d$-dimensional weighted simplicial complex.
    For any $0 \leq k \leq d$,
    \[ \lambda_2(\DownW_{k}) = \lambda_2(\UpW_{k-1}) \le 1 - \frac{1}{k+1} \prod_{j = - 1}^{k-2} (1 - \gamma_j).\]
\end{restatable}

Using an inductive argument as in~\cite[Theorem 3.3]{AnariLOV18}, we can prove a more general statement about the entire range of eigenvalues.

\begin{restatable}{theorem}{maingeneral}\label{thm:main-general}
    Let $(X, \Pi)$ be a pure $d$-dimensional weighted simplicial complex. 
    Then, for any $0 \leq k \leq d-1$ and for any $-1 \leq r \leq k$,
    the matrix $\UpW_k$ has at most $|X(r)|$ eigenvalues with value strictly greater than
    \[ 1 - \frac{1}{k + 2} \prod_{j = r}^{k-1} (1 -\gamma_j).\]
\end{restatable}
Note that \cref{thm:main} is a special case of \cref{thm:main-general} where
$r=-1$ (recall that $X(-1)=\{\emptyset\}$ and so $|X(-1)|=1$). 
Further, \cref{thm:main} can only prove that $\lambda_2(\DownW_d) \leq 1-\frac{1}{d+1}$.
We observe that this bound is almost tight.

\begin{proposition}\label{prop:nonexp}
    Let $X$ be a $d$-dimensional simplicial complex. 
    Let $n = |X(0)|$.
    Suppose $2 (d + 1) \le n$.
    Then $\lambda_2(\DownW_{d}) = \lambda_2(\UpW_{d-1}) \ge 1 - \frac{2}{d + 1}$.
\end{proposition}

Before we prove \cref{thm:main} and \cref{thm:main-general}, 
we present two corollaries of \cref{thm:main}. 

Combining with Oppenheim's \cref{cor:oppenheim}, 
\cref{thm:main} provides a bound on the second eigenvalue of the $d$-th down-up walk based only on the maximum second eigenvalue of the graphs in dimension $d-2$.
This will be useful in \cref{sec:sampling}.

\begin{restatable}{corollary}{maincooked}\label{cor:main-cooked}
    Let $(X, \Pi)$ be a pure $d$-dimensional weighted simplicial complex. 
    For any $0 \leq k \leq d$,
    suppose $\gamma_k \leq \frac{1}{k+1}$ and $\gamma_j < 1$ for $-1 \leq j \leq k-2$,
    then
    \[ 
    \lambda_2(\DownW_k) = \lambda_2(\UpW_{k - 1}) \le 1 - \frac{1}{(k+1)^2}.
    \]
\end{restatable}
\begin{proof}
    Since $\gamma_{k-2} \le \frac{1}{k+1}$ and $\gamma_j < 1$ for $-1 \leq j \leq k-2$,
    it follows from Oppenheim's \cref{cor:oppenheim} that for any $-1 \leq j \leq k-3$, 
    \[
        \gamma_j \leq \frac{\gamma_{k-2}}{1-(k-2-j)\cdot \gamma_{k-2}}
    \leq \frac{\frac{1}{k+1}}{1-\frac{k-2-j}{k+1}}  = \frac{1}{j+3}.
\]
    Therefore, by \cref{thm:main},
    \begin{align*}
        \lambda_2(\DownW_k)
        \leq 1 - \frac{1}{k + 1} \prod_{j = -1}^{k-2} (1 - \gamma_j)
        \leq 1 - \frac{1}{k + 1} \prod_{j = -1}^{k-2} \frac{j+2}{j+ 3}
        = 1 - \frac{1}{(k+1)^2}.
    \end{align*}
\end{proof}

\cref{thm:main} implies the following result for longer random walks on local-spectral expanders.

\begin{restatable}{corollary}{abwalk}\label{cor:abwalk}
    Let $(X,\Pi)$ be a pure $d$-dimensional weighted simplicial complex. 
    Let $0 \le a < b \le d -1$.
    If $X$ is a $\gamma$-local-spectral expander, then
    \[ \lambda_2(\UpW_{a, b}) \le (1 + \gamma)^{b - a} \cdot
    \frac{a+1}{b+1}.\]
\end{restatable}

The rest of this section is organized as follows.
We will first prove \cref{thm:main} in \cref{ss:main},
then \cref{thm:main-general} in \cref{ss:main-general},
then \cref{cor:abwalk} in \cref{ss:abwalk},
and finally \cref{prop:nonexp} in \cref{ss:nonexp}.

\subsection{Proof of \cref{thm:main}} \label{ss:main}

The key lemma in proving \cref{thm:main} is the following result that
quantifies a spectral bound on the difference of the $k$-th non-lazy up-down walk and the $k$-th down-up walk in terms of the second eigenvalue of the links at dimension $k-1$.

\begin{restatable}{lemma}{updownrel}\label{lem:updownrel}
    Let $(X, \Pi)$ be a pure $d$-dimensional weighted simplicial complex. 
    For any $0 \leq k \leq d-1$,
    \[ 
    \NUpW_k - \DownW_k \preceq_{\Pi_k} \gamma_{k-1} \cdot \parens*{ \Ide - \DownW_k}. 
    \]
\end{restatable}

The proof of \cref{lem:updownrel}, 
will closely follow the proof of \cite[Theorem 5.5]{DiksteinDFH18},
where they prove the weaker inequality
\begin{equation}
    \NUpW_k - \DownW_k \preceq_{\Pi_k} \gamma_{k-1} \cdot \Ide.\label{eq:updownrelw}
\end{equation}
We remark that a similar statement was also used in~\cite{KaufmanO18} for proving \cref{t:KO}.

We will first show how \cref{lem:updownrel} implies \cref{thm:main} by an inductive argument.

\begin{proof}[Proof of \cref{thm:main} from \cref{lem:updownrel}]
    We prove \cref{thm:main} by induction on $k$.
    The base case is when $k = 0$, where $\DownW_0 = \one \Pi_0^\top$ is a rank one matrix and so $\lambda_2(\DownW_0) \leq 0$, and hence \cref{thm:main} trivially holds.

    For the induction step, suppose we have
    \begin{equation}
        \lambda_2(\DownW_{j+1}) = \lambda_2(\UpW_j) \le 1 - \frac{1}{j+2} \prod_{i=-1}^{j-1} (1 - \gamma_i).
        \label{eq:ind-hyp}\tag{induction hypothesis}
    \end{equation}

    Since $\DownW_{j+ 1} = \Up_{j}\Dee_{j+1}$ and $\UpW_j = \Dee_{j+1}\Up_j$ have the same non-zero eigenvalues with the same multiplicity by \cref{fac:simple-la}, we only need to prove the statement for $\UpW_{j+1}$.
    By \cref{lem:updownrel},
    \[
        \NUpW_{j+1} \preceq_{\Pi_{j+1}} \gamma_{j} \cdot \Ide + (1 - \gamma_{j}) \DownW_{j+1}
    \]
    It follows from \cref{fac:simple-loew} that
    \[ 
    \lambda_2(\NUpW_{j+1}) 
    \leq \gamma_j + (1 - \gamma_j) \cdot \lambda_2(\DownW_{j+1}) 
    \leq 1 - \frac{1}{j+2} \prod_{i = -1}^{j} (1 - \gamma_i),
    \]
    where the last equality is by plugging in the \ref{eq:ind-hyp}. 
    The theorem now follows from the definition of \ref{eq:nupw-def}, i.e.
    \[ 
    \NUpW_{j+1} = \frac{j+3}{j+2} \parens*{\UpW_{j+1} - \frac{1}{j+3} \Ide } 
    \quad \Longleftrightarrow \quad 
    \UpW_{j+1} = \frac{j+2}{j+3} \cdot \NUpW_{j+1} + \frac{1}{j+3} \Ide.
    \]
    Therefore,
    \[ \lambda_2(\UpW_{j+1}) = \frac{j+2}{j+3} \cdot \lambda_2(\NUpW_{j+1}) + \frac{1}{j+3} \le 1 - \frac{1}{j+3} \prod_{i=-1}^{j} (1- \gamma_i),\]
    and this proves the induction step.
\end{proof}

\subsubsection{Proof of \cref{lem:updownrel}}

The proof of \cref{lem:updownrel} will rest on few useful identities
established in \cite{KaufmanO18, DiksteinDFH18}, which can be obtained through the ``Garland Method'', which decomposes the higher order random walk matrices into the random walk matrices of the links.

In the following, given $\eff \in \RR^{X(k)}$ and $\alpha \in X(k-1)$, we use $\eff_{\alpha}$ to denote the restriction of $\eff$ to the entries in $\{ \alpha \cup \{x\} \mid x \in X_{\alpha}(0) \}$.
And recall that $\Jay_\alpha$ is the \ref{eq:J} defined in \cref{ss:local-spectral}

\begin{restatable}{lemma}{garland}\label{lem:garland}
    Let $(X, \Pi)$ be a pure $d$-dimensional weighted simplicial complex.  
    For all $\eff \in \RR^{X(j)}$ the following hold,
    \begin{enumerate}
        \item $\langle \eff, \Ide \eff\rangle_{\Pi_j}  = \Exp_{\alpha \sim \Pi_{j-1}}\norm{\eff_\alpha}^2_{\Pi_0^\alpha} = \Exp_{\alpha \sim \Pi_{j-1}}\langle \eff_\alpha, \eff_\alpha\rangle_{\Pi_0^\alpha}$,
        \item $\langle \eff, \DownW_j \eff\rangle_{\Pi_j} = \Exp_{\alpha \sim \Pi_{j-1}} \norm{\Jay_\alpha \eff_\alpha}_{\Pi_0^\alpha}^2 = \Exp_{\alpha \sim \Pi_{j-1}} \langle \eff_\alpha, \Jay_\alpha \eff_\alpha\rangle_{\Pi_0^\alpha}$,
        \item $\langle \eff, \NUpW_j \eff\rangle_{\Pi_j} = \Exp_{\alpha \sim
            \Pi_{j-1}} \langle \eff_\alpha, \Emm_\alpha \eff_\alpha\rangle_{\Pi_0^\alpha}$.
    \end{enumerate}
\end{restatable}

We will provide a proof of \cref{lem:garland} in \cref{ss:garland} for completeness.
We are ready to prove \cref{lem:updownrel}. 
\begin{proof}[Proof of \cref{lem:updownrel}]
    Let $\eff \in \RR^{X(j)}$ be arbitrary. 
    By Items (2) and (3) in \cref{lem:garland}, we write
    \[
        \langle \eff, (\NUpW_j - \DownW_j)\eff\rangle_{\Pi_j} 
    = \Exp_{\alpha \sim \Pi_{j-1}}\sqbr*{\langle \eff_\alpha, (\Emm_\alpha - \Jay_\alpha) \eff_\alpha \rangle_{\Pi_0^\alpha}}.
\]
    Notice that since $\Emm_\alpha$ is a row-stochastic matrix (with top eigenvector $\one$) and
    the matrix $\Jay_\alpha$ is the projector to its top eigenspace. 
    Since both $\Emm_\alpha$ and $\Jay_\alpha$ are self-adjoint with respect
    to the inner-product defined by $\Pi_0^\alpha$ (see \cref{ss:local-spectral}), it follows that
    \[ \Emm_\alpha - \Jay_\alpha \preceq_{\Pi_0^{\alpha}} \lambda_2(\Emm_\alpha) \cdot \Ide.\]
    Moreover, since the matrix $\Emm_\alpha - \Jay_\alpha$ is only supported
    on the subspace perpendicular to $\one$, writing $\eff_\alpha^{\perp \one}$ for the component of $\eff_\alpha$ that is perpendicular to $\one$, we have
    \[ 
    \langle \eff_\alpha, (\Emm_\alpha - \Jay_\alpha)
    \eff_\alpha\rangle_{\Pi_0^\alpha} = \langle \eff_\alpha^{\perp \one},
    (\Emm_\alpha - \Jay_\alpha) \eff_\alpha^{\perp \one}\rangle_{\Pi_0^\alpha}.
    \]
    As $\Jay_\alpha$ is the \ref{eq:J} we have, $\eff_\alpha^{\perp
    \one} = (\Ide - \Jay_\alpha) \eff_\alpha$ and thus
    \begin{align}
        \langle \eff_\alpha, (\Emm_\alpha- \Jay_\alpha) \eff_\alpha\rangle
        &~\le~\lambda_1( \Emm_\alpha - \Jay_\alpha) \cdot \norm{\eff_{\alpha}^{\perp
        \one}}_{\Pi_0^\alpha}^2
        ~\le~\lambda_2(\Emm_\alpha) \cdot \norm{(\Ide -
        \Jay_\alpha)\eff_\alpha}_{\Pi_0^\alpha}^2,\label{eq:boundusefulb}
    \end{align}
    where the first inequality is by the \ref{eq:var-form}.    
    Therefore,
    \begin{align*}
        \langle \eff, (\NUpW_j - \DownW_j) \eff\rangle_{\Pi_j} &~=~
        \Exp_{\alpha \sim \Pi_{j-1}} \sqbr*{ \langle \eff_\alpha, (\Emm_\alpha -
        \Jay_\alpha) \eff_\alpha\rangle_{\Pi_0^\alpha} },&&\textrm{ (by Items (2)
        and (3) in \cref{lem:garland})}\\
        &~\le~ \Exp_{\alpha \sim \Pi_{j-1}}\sqbr*{ \lambda_2(\Emm_\alpha) \cdot
        \norm{(\Ide - \Jay_\alpha) \eff_\alpha}^2_{\Pi_0^\alpha} },&&\textrm{
            (by \cref{eq:boundusefulb})}\\
        &~\le~\gamma_{j-1} \cdot \Exp_{\alpha \sim \Pi_{j-1}} \sqbr*{\norm*{(\Ide - \Jay_\alpha)\eff_\alpha}_{\Pi_0^\alpha}^2},\\
        &~=~ \gamma_{j-1} \cdot \Exp_{\alpha \sim \Pi_{j-1}} \sqbr*{ \langle \eff_\alpha, (\Ide - \Jay_\alpha)\eff_\alpha\rangle_{\Pi_0^\alpha}},&&
        \textrm{ (by $\langle \Jay_\alpha \eff_\alpha, \Jay_\alpha \eff_\alpha \rangle_{\Pi_0^\alpha} = \langle \eff_\alpha, \Jay_\alpha \eff_\alpha \rangle_{\Pi_0^\alpha}$)}\\
        &~=~ \gamma_{j- 1} \cdot \langle \eff, (\Ide - \DownW_j)
        \eff\rangle_{\Pi_j}&&\textrm{ (by Items (1) and (2) in
        \cref{lem:garland})}.
    \end{align*}
    This proves $\NUpW_j - \DownW_j \preceq_{\Pi_j} \gamma_{j-1} (\Ide - \DownW_j)$.
\end{proof}

\subsubsection{Proof of \cref{lem:garland}} \label{ss:garland}

Here we provide a proof of \cref{ss:garland} for completeness.
These arguments are from \cite{KaufmanO18, DiksteinDFH18}.

\begin{proof}
    Item (1) can be proven from the identity
    \[
        \Pi_j(\beta)
    = \sum_{\alpha \in X(j-1), x \in X(0), \atop\alpha\cup x =\beta}  \frac{\Pi_j(\alpha \cup \{x\})}{k+1}
    = \sum_{\alpha \in X(j-1), x \in X(0), \atop\alpha\cup x =\beta} \Pi_{j-1}(\alpha) \cdot \Pi_0^\alpha(x),
\]
    where the last equality is by \cref{eq:link-def} that $\Pi_0^\alpha(x) =
    \frac{ \Pi_j(\alpha \cup \{x\}) } {(j+1) \cdot \Pi_{j-1}(\alpha)}$.
    Then,
    \begin{align*}
        \langle \eff, \Ide \eff\rangle_{\Pi_j} &~=~\sum_{\beta \in X(j)} \Pi_j(\beta) \cdot \eff(\beta)^2\\
        &~=~ \sum_{\beta \in X(j)} \sum_{\alpha \in X(j-1), x \in X(0),
        \atop\alpha\cup x =\beta} \Pi_{j-1}(\alpha) \cdot \Pi_0^\alpha(x) \cdot
        \eff_\alpha(x)^2\\
        &~=~ \sum_{\alpha \in X(j-1)} \Pi_{j-1}(\alpha) \cdot \sum_{x \in X_\alpha(0)} \Pi_0^\alpha(x) \cdot \eff_\alpha(x)^2\\
        &~=~ \Exp_{\alpha \sim \Pi_{j-1}} \langle \eff_\alpha, \eff_\alpha\rangle_{\Pi_0^\alpha}.
    \end{align*}
    Item (2) follows by appealing to the definition of the down-up walk that $\DownW_j = \Up_{j-1} \Dee_j$, and so
    \[ 
    \langle \eff, \DownW_j \eff\rangle_{\Pi_j} 
    = \langle \eff, \Up_{j-1}\Dee_j \eff\rangle_{\Pi_j} 
    = \langle \Dee_j \eff, \Dee_j \eff\rangle_{\Pi_{j-1}}. 
    \]
    By the definition of the \ref{eq:down-def} and $\Pi_0^\alpha(x) =
    \frac{\Pi_j(\alpha \cup \{x\}) }{(j+1) \cdot \Pi_{j-1}(\alpha)}$ from \cref{eq:link-def}, 
    it follows that 
    $[\Dee_j \eff](\alpha) 
    = \sum_{x \in X_{\alpha}(0)} \Pi_0^\alpha(x) \cdot \eff(\alpha \cup \{x\})
    = \Exp_{x \sim \Pi_0^\alpha} \eff_\alpha(x)
    $ 
    and thus
    \[ 
    \langle \eff, \DownW_j \eff\rangle_{\Pi_j} 
    = \sum_{\alpha \in X(j-1)} \Pi_{j-1}(\alpha) \left(\Exp_{x \sim \Pi_0^\alpha} \eff_\alpha(x) \right)^2
    = \Exp_{\alpha \sim \Pi_{j-1}}\sqbr*{ \parens*{ \Exp_{x \sim \Pi_0^\alpha} \eff_\alpha(x)}^2 }.
    \]
    Observing that $\Jay_\alpha\eff_\alpha = \one \cdot \Exp_{x \sim \Pi_0^\alpha}
    \eff_\alpha(x)$ by the definition of the \ref{eq:J}
    and therefore $\norm{\Jay_\alpha \eff_\alpha}^2_{\Pi_0^\alpha} 
    = \parens*{\Exp_{x \sim \Pi_0^\alpha} \eff_\alpha(x)}^2$.
    Hence, Item (2) follows as
    \[ \langle \eff, \DownW_j \eff\rangle_{\Pi_j} = \Exp_{\alpha \sim \Pi_{j-1}} \norm{\Jay_\alpha \eff_\alpha}^2_{\Pi_0^\alpha} = \Exp_{\alpha \sim \Pi_{j-1}} \langle \eff_\alpha, \Jay_\alpha \eff_\alpha\rangle_{\Pi_0^\alpha},\]
    where we used that $\Jay_\alpha$ is an orthogonal projection and so $\langle \Jay_\alpha \eff_\alpha, \Jay_\alpha \eff_\alpha \rangle_{\Pi_0^\alpha} = \langle \eff_\alpha, \Jay_\alpha \eff_\alpha \rangle_{\Pi_0^\alpha}$.

    For Item (3), by the definition of $\UpW_j = \Dee_{j+1} \Up_j$ and the definition of \ref{eq:up-def},
    \begin{align*}
        \langle \eff, \UpW_j \eff\rangle_{\Pi_j} ~=~\langle \Up_{j}\eff,
        \Up_{j} \eff\rangle_{\Pi_{j+1}}
        ~=~ \sum_{\beta \in X(j+1)} \Pi_{j+1}(\beta) \cdot \sum_{x, y \in \beta}
        \frac{1}{|\beta|^2} \eff(\beta\backslash x) \eff(\beta \backslash y).
    \end{align*}
    Now, by the definition of \ref{eq:nupw-def}, we see that
    \begin{align*}
        \langle \eff, \NUpW_j \eff\rangle_{\Pi_j} 
        &~=~\frac{j+2}{j+1} \cdot \langle \eff, \UpW_j \eff\rangle_{\Pi_j} 
        - \frac{1}{j+1} \langle \eff, \eff\rangle_{\Pi_j},
        \\
        &~=~\sum_{\beta \in X(j+1) }  
        \frac{\Pi_{j+1}(\beta)}{|\beta| \cdot (|\beta| - 1)} \sum_{x, y \in \beta}
        \eff(\beta\backslash x) \eff(\beta\backslash y) 
        - \frac{1}{j+1} \sum_{\alpha \in X(j)} \Pi_j(\alpha) \eff(\alpha) \eff(\alpha),
    \end{align*}
where we got the second inequality using $|\beta| = j+2$. 
Now, notice that sampling $\alpha \sim \Pi_j$ is the same as first sampling $\beta \sim \Pi_{j+1}$ and then sampling $x \sim \beta$ uniformly and considering $\beta\backslash x$, so we can get by \cref{eq:onestep} that
\begin{align*}
    \langle \eff, \NUpW_j \eff\rangle_{\Pi_j} 
    &~=~\sum_{\beta \in X(j+1)}  \sum_{x, y \in \beta}
    \frac{\Pi_{j+1}(\beta)}{|\beta| \cdot (|\beta| - 1)} \eff(\beta\backslash x)
    \eff(\beta\backslash y) -  \frac{1}{(j+1)} \sum_{\beta \in
    X(j+1)} \Pi_{j+1}(\beta) \cdot \sum_{x \in \beta} \frac{\eff(\beta\backslash x) \eff(\beta \backslash x)}{j+2}, 
    \\
    &~=~\sum_{\beta \in X(j+1)} \Pi_{j+1}(\beta) \sum_{\set*{x, y} \in \beta}
    \frac{1}{\binom{|\beta|}{2}} \eff(\beta\backslash x)
    \eff(\beta\backslash y)
\end{align*}
where we have obtained the last inequality by using $|\beta| = j + 2$ and noticing that the sum kills the diagonal terms. 
Using $\tau = \beta\backslash \set{x, y}$ and the identity
$\frac{\Pi_{j+1}(\beta)}{ \binom{|\beta|}{2} } = \Pi_{j-1}(\tau) \cdot \Pi_1^\tau(\{x, y\})$ from \cref{eq:link-def}, we can rewrite it as
\begin{align*}
    \langle \eff, \NUpW_j \eff\rangle_{\Pi_j} 
    &~=~\sum_{\beta \in X(j+1)} \sum_{\set*{x, y} \in \beta}
    \Pi_1^\tau(\{x,y\}) \cdot  \Pi_{j-1}(\tau) \cdot \eff(\tau \cup x) \eff(\tau \cup y)
    \\
    &~=~\sum_{\tau \in X(j-1)} \Pi_{j-1}(\tau) \sum_{\{x,y\} \in X_\tau(1)} \Pi_1^\tau(\{x,y\}) \eff(\tau \cup x) \eff(\tau \cup y).
\end{align*}

On the other hand, using the equation
\begin{equation*}
    \langle \eff, \Emm_\tau \eff\rangle_{\Pi_0^\tau} = \sum_{x \in
    X_\tau(0)} \Pi_0^\tau(x) \cdot \eff(x) \cdot [\Emm_\tau \eff](x) = \sum_{\set*{x, y} \in X_\tau(1)} \eff(x) \eff(y) \cdot \Pi_1^\tau(x, y).
\end{equation*}
where we use $\Emm_{\tau}(x,y) = \frac{\Pi_1^\tau(x,y)}{2\Pi_0^\tau(x)}$ from \cref{ss:local-spectral},
we can also write
\begin{align*}
    \Exp_{\tau \sim \Pi_{j-1}}\sqbr*{\langle \eff_\tau, \Emm_\tau,
        \eff_\tau\rangle_{\Pi_{0}^\tau} } 
    &~=~\sum_{\tau \in X(j-1)} \Pi_{j-1}(\tau) \cdot \sum_{\set*{x,y} \in X_\tau(1)} \Pi^{\tau}_1(\{x, y\}) \cdot \eff_\tau(x) \cdot \eff_\tau(y),
\end{align*}
and this proves $\langle \eff, \NUpW_j \eff\rangle_{\Pi_j}  = \Exp_{\tau \sim \Pi_{j-1}}\sqbr*{\langle \eff_\tau, \Emm_\tau,
\eff_\tau\rangle_{\Pi_{0}^\tau} }$.
\end{proof}

\subsection{Proof of \cref{thm:main-general}} \label{ss:main-general}

We will prove \cref{thm:main-general} about the entire spectrum of the higher order random walks.
\maingeneral*
\begin{proof}
    We prove by induction on $k$.
    The base case is when $k = 0$, 
    where $\UpW_0 = \frac{1}{2} \Emm_\varnothing + \frac{1}{2} \Ide$. 
    The claim states that we have at most $|X(-1)| = 1$ eigenvalue that is strictly greater than $\frac{1}{2} + \frac{\gamma_0}{2} = \frac{1}{2} + \frac{\lambda_2(\Emm_\varnothing)}{2}$ (which is true by the definition of $\lambda_2(\Emm_\varnothing)$),
    and there are at most $|X(0)|$ eigenvalues strictly greater than $1/2$ (which is true as the $\Emm_\varnothing$ is of rank at most $|X(0)|$). 

    For the induction step, suppose that there exists some $j \geq 1$
    such that the claim of the theorem is true for all $-1 \leq r \leq j$. 
    By \cref{fac:simple-la}, $\DownW_{j+1}$ and $\UpW_j$ have the same non-zero
    eigenvalues.
    By \cref{lem:updownrel}, $\NUpW_{j+1} \preceq_{\Pi_{j+1}} \gamma_{j} \Ide + (1-\gamma_{j}) \DownW_j$, and thus for $-1 \leq r \leq j$ the matrix $\NUpW_{j+1}$ has at most $|X(r)|$ eigenvalues with value greater than
    \[ 
    \gamma_{j} + (1 - \gamma_j) \cdot \parens*{ 1 - \frac{1}{j + 2} \prod_{i=r}^{j-1} (1 - \gamma_i) }
    = 1 - \frac{1}{j+2}\prod_{i=r}^{j} (1 - \gamma_i).
    \]
    Using the definition of the \ref{eq:nupw-def}, 
    we have that for $-1 \leq r \leq j$, 
    $\UpW_{j+1}$ has at most $|X(r)|$ eigenvalues with value greater than
    \[ \frac{j+2}{j+3} \left( 1 - \frac{1}{j+2}\prod_{i=r}^{j} (1 - \gamma_i) \right) + \frac{1}{j+3}
    = 1 - \frac{1}{j+3} \prod_{i=r}^{j} (1 - \gamma_i).
    \]
    For $r = j+1$, it is trivial that there at most $|X(j+1)|$ eigenvalues greater than $\frac{j+1}{j+2}$ since $\UpW_{j+1}$ is an operator of rank at most $|X(j+1)|$.
\end{proof}

\subsection{Proof of \cref{cor:abwalk}}\label{ss:abwalk}

\abwalk*
We will use two basic facts in the proof.
\begin{fact}\label{fac:split}
    Let $\Emm_1 \in \RR^{V \times U}$ and $\Emm_2 \in \RR^{U \times W}$ be two
    row-stochastic matrices. Then, we have $\sigma_2(\Emm_1 \cdot \Emm_2) \le
    \sigma_2(\Emm_1) \cdot \sigma_2(\Emm_2).$
\end{fact}
\begin{fact}[Bernoulli's Inequality]\label{fac:bernoulli}
    Let $x \ge - 1$ and $r \ge 1$ be real numbers. Then, $(1 + x)^r \ge 1 + r
    \cdot x$.
\end{fact}
\begin{proof}
    Recall that $\UpW_{a, b} = \Dee_{a + 1} \cdots \Dee_b \cdot \Up_{b-1} \cdots \Up_{a}$.
    As $\Up_i^* = \Dee_{i+1}$, it can be observed that $\UpW_{a, b}$ is positive semi-definite and therefore, $\sigma_2(\UpW_{a, b}) = \lambda_2(\UpW_{a, b})$.
    By \cref{fac:split},
    \[  \lambda_2(\UpW_{a, b}) = \sigma_2(\UpW_{a, b}) \le
    \sigma_2(\Dee_{a + 1}) \cdots \sigma_2(\Dee_{b}) \cdot
    \sigma_2(\Up_{b - 1}) \cdots \sigma_2(\Up_{a}).\]
    Notice that as $\UpW_j = \Dee_{j+1} \Up_j$ and $\Dee_{j+1}^* = \Up_j$, we
    have $\lambda_2(\UpW_j) = \sigma_2(\Up_j) \cdot \sigma_2(\Dee_{j+1})$.
    Thus, by rearranging we obtain,
    \begin{align*}
        \lambda_2(\UpW_{a, b}) &~\le~\prod_{j = 0}^{b - a - 1}
        \lambda_2(\UpW_{a + j}),\\
        &~\le~\prod_{j = 0}^{b- a - 1} \parens*{ 1 - \frac{ (1 -
        \gamma)^{a + j + 1}}{a + j + 2}}, && \textrm{ (by \cref{thm:main})}\\
        &~\le~\prod_{j=0}^{b - a - 1} \parens*{ 1 - \frac{1 - (a + j+ 1)
        \cdot \gamma}{a + j + 2}}, && \textrm{ (by \cref{fac:bernoulli})}\\
        &~=~\prod_{j = 0}^{b - a - 1}\parens*{ (1 + \gamma) \frac{a + j +1}{a + j +
        2}}.
    \end{align*}
    By cancellations in the telescoping product, we have
    \[ \lambda_2(\UpW_{b, a}) \le (1 + \gamma)^{b - a} \cdot \frac{a
    + 1}{b + 1}.\]
\end{proof}

\subsection{Proof of \cref{prop:nonexp}} \label{ss:nonexp}

We first recall some basic definitions and results from spectral graph theory.
Let $\Emm \in \RR^{V \times V}$ be a reversible Markov chain with stationary
distribution $\Pi$. We will write
$(X_t)_{t \ge 0}$ for the random variable describing the state of this Markov
chain.
The conductance $\Phi(S)$ of a set $S \subset V$  is the probability that a
random step of the Markov chain leaving the set $S$ conditioned on having
started from a random state in $S$, i.e.~
\[ \Phi(S) = \Pr[ X_1 \not\in S \mid X_0 \in S] = \sum_{x \in S}
\frac{\Pi(x)}{\Pi(S)} \Pr[X_1 \not\in S \mid X_0 = x] .\]
We recall that the conductance of the chain $\Emm$ is defined to be,
\[ \Phi(\Emm) = \min_{S \subset V,\atop \Pi(S) \le 1/2} \Phi(S).\]
The Alon-Milman-Cheeger inequality tells that the parameter $\Phi(M)$
is closely related to the parameter $\lambda_2(\Emm)$. 
\begin{theorem}[\cite{AlonMilman, Cheeger}]\label{thm:cheeger}
    Let $\Emm$ be a row-stochastic matrix describing a reversible Markov chain. Then,
    \[ \frac{1 - \lambda_2(\Emm)}{2} \le \Phi(\Emm) \le \sqrt{ \frac{1 -
    \lambda_2(\Emm)}{2}}.\]
\end{theorem}
\begin{proof}[Proof of \cref{prop:nonexp}]
    It is clear that there should exist a vertex $v \in X(0)$ such that
    $\Pi_0(v) \le \frac{1}{|X(0)|} = \frac{1}{n}$.

    We consider the set $A_v \subset X(d)$ consisting of all faces in $X(d)$
    containing the vertex $v$, i.e.~$A_v = \set*{\beta \in X(d): v \in \beta}$. Note that,
    \begin{align*}
        \Pi_d(A_v) &~=~(d+1) \cdot \Pi_0(v), && \textrm{ (by using
        \cref{eq:onestep} repeatedly)},\\
        &~\le~\frac{d+1}{n}, && \textrm{ (by using $\Pi_0(v) \le
        \frac{1}{n}$)},\\
        &~\le~\frac{1}{2}. &&\textrm{ (by using $2(d+1) \le n$)}
    \end{align*}
    By \cref{thm:cheeger},
    \[ \frac{1 - \lambda_2(\DownW_d)}{2} \le \min_{S: \Pi_d(S) \le 1/2} \Phi(S)
    \le \Phi(A_v). \]
    We recall that the random-walk $\DownW_d$ starting from a face $\beta \in
    X(d)$ first picks an index $i \sim \beta$ uniformly at random, and
    then picks some face $\beta' \supset (\beta \backslash i)$ with
    probability proportional to $\Pi_d(\beta')$. If $\beta \in A_v$ the
    only way we leave $A_v$ in a single step is
    when the index $i$ we pick from $\beta$ is $v$,
    which happens with probability $1/|\beta| = 1/(d+1)$. 

    Writing $(X_t)_{t\ge 0}$
    for the state of the random walk, this means for any $\beta \in A_v$ we have
    \[ \Pr[X_1 \not\in A_v \mid X_0 = \beta] \le \frac{1}{d+1}.\]
    It follows that
    \[ \frac{1  - \lambda_2(\DownW_d)}{2} \le \Phi(A_v) = \sum_{\beta \in A_v} \frac{\Pi_d(\beta)}{\Pi_d(A_v)} \cdot
    \Pr[X_1 \not\in A_v \mid X_0 = \beta] \le \max_{\beta \in A_v} \Pr[X_1 \not\in
    A_v \mid X_0 = \beta] \le \frac{1}{d+1}.\]
    Solving the expression for $\lambda_2(\DownW_d)$ proves the proposition.
\end{proof}

\section{Analyzing Mixing Times of Markov Chains}\label{sec:sampling}

In this section, we will use \cref{cor:main-cooked} to analyze Markov chains for sampling independent sets of a graph of fixed size and sampling common independent sets of two partition matroids.

\subsection{Independent Sets}\label{ss:is-main}

Let $G = (V,E)$ be a graph.
A subset of vertices $S \subset V$ is called an independent set if $uv \notin E$ for every pair $u,v \in S$.
We are interested in the problem of sampling a uniformly random independent set of size $k$.
We will analyze a natural Markov chain for the problem by analyzing the down-up walk of a corresponding simplicial complex.

Define the $(k-1)$-dimensional simplicial complex $I_{G, k}$ of $G=(V,E)$ as 
\[ I_{G, k} = \set*{ S \subset V : |S| \le k \textrm{ and } S \textrm{ is
independent}},\]
the complex consisting of all independent sets in $G$ of cardinality at most $k$.
We endow $I_{G, k}$ with the uniform distribution $\Pi_{k-1}$ on $I_{G, k}(k-1)$, i.e.~the set of independent sets of size $k$. 
We simply write $I_{G, k}$ for the weighted simplicial complex $(I_{G, k}, \Pi_{k-1})$.

The $(k-1)$-th down-up walk $\DownW_{k-1}$ on $I_{G, k}$ corresponds to a natural Markov chain to sample independent sets of size $k$.
It is known that this Markov chain is fast mixing when $k \leq
\frac{|V|}{2\Delta+1}$ using coupling techniques~\cite{BubleyD97,MitzenmacherUpfal05}.
The main result in this subsection is the following improved bound using higher order random walks on simplicial complexes.

\IS*

It is well-known that $|\lambda_{\min}(\Aye_G)| \leq \Delta$ for a graph with
maximum degree $\Delta$, and so Theorem~\ref{t:IS} recovers the previous result
that the Markov chain is fast mixing if $k \leq \frac{|V|}{2\Delta}$.
There are various graph classes with $|\lambda_{\min}(\Aye_G)|$ smaller than $\Delta$, and \cref{t:IS} allows us to sample larger independent sets.
For example, it is known that $|\lambda_{\min}(\Aye_G)| \leq O(\sqrt{\Delta})$ for planar graphs and more generally for graphs with bounded arboricity~\cite{Hayes06}, and also for random graphs and more generally for two-sided expander graphs~\cite{HooryLW06}.

Using the simple bound $\min_{S \in I_{G, k}(k-1)} \Pi_{k-1}(S) \ge n^{-k}$
as $\Pi_{k-1}$ is the uniform distribution, the following mixing time result follows from \cref{thm:spec-mix-bd}.

\begin{restatable}{corollary}{indsetalg}\label{cor:independentset-alg}
Let $G = (V, E)$ be a graph with maximum degree $\Delta$ and let $\Aye_G$ be the adjacency matrix of $G$.
For any $k \le n/(\Delta +|\lambda_{\min}(\Aye_G)|)$,
the down-up walk $\DownW_{k-1}$ on the simplicial complex $I_{G, k}$ 
samples a random independent set of $G$ of size $k$ whose distribution is 
\ref{eq:eps-close} to the uniform distribution on all independent sets of size $k$ in
    \[ T(\ee, \DownW_{k-1}) \le  k^2 \cdot \parens*{\log\parens*{\frac{1}{\ee}} + k \cdot \log n}\]
    many time steps.
\end{restatable}

This implies a polynomial time algorithm to approximately sample a uniform random independent set and also a FPRAS for approximately counting the number of independent set of size $k$ for $k \le \frac{n}{\Delta +|\lambda_{\min}(\Aye_G)|}$.

\subsubsection{Proof of \cref{t:IS}} \label{ss:proofIS}

The plan is to use \cref{cor:main-cooked} to prove \cref{t:IS}.
To apply \cref{cor:main-cooked}, we need to prove that:
\begin{enumerate}
\item $I_{G, k}$ is a pure simplicial complex.  
It is a simple exercise that this complex is pure when $k \leq \frac{n}{\Delta+1}$.
\item For each $S \in I_{G, k}$ with $|S| \leq k-2$, 
the random walk matrix $\Emm_S$ of the graph $G_S$ of the link $(I_{G,k})_S$ satisfies $\lambda_2(\Emm_S) < 1$.
This is proved in \cref{lem:IS-connected}.
\item For each $S \in I_{G, k}$ with $|S| = k-2$,
the random walk matrix $\Emm_S$ of the graph $G_S$ satisfies $\lambda_2(\Emm_S) \leq 1/k$.
This is proved in \cref{lem:IS-top-link}.
\end{enumerate}
Assuming the three items are proven, \cref{t:IS} follows immediately from \cref{cor:main-cooked}.
We will prove the second item in \cref{ss:IS-connected}
and the third item in \cref{ss:IS-top-link}.

\subsubsection{Proof of \cref{lem:IS-connected}} \label{ss:IS-connected}

Let $H_S = (V_S, E_S)$ be the underlying support graph of $G_S$ of the link $(I_{G,k})_S$, i.e. $G_S$ without edge weights.
Let $\Emm_S$ be the random walk matrix of $G_S$ as defined in \cref{ss:local-spectral}.
Note that $\lambda_2(\Emm_S) < 1$ if and only if $H_S$ is connected.

We introduce some notation to describe $H_S$.
We write $N_G[S]$ as the union of $S$ and the set of vertices which are connected to a vertex in $S$ in $G$, i.e.~
\[ N_G[S] = S \cup \set*{v : \textrm{there exists some }uv \in E(G) \textrm{ such that } u \in S}.\]
For a subset of vertices $S \subset V(G)$, we write $\conj{S} = V(G) \setminus S$ for the complement of $S$ in $G$, and $G[S]$ for the induced subgraph of $G$ on $S$.
For a graph $H$, we write $\conj{H}$ for the complement graph of $H$.

Recall that a vertex $v$ is in $V_S$ if and only if $S \cup \{v\}$ is an independent set in $G$ of size $|S|+1$,
and so $V_S$ is exactly $V - N_G[S] = \conj{N_G[S]}$.
Two vertices $u,v \in V_S$ have an edge in $H_S$ if and only if $S \cup \{u,v\}$ is an independent set in $G$ of size $|S|+2$,
and so $uv \in E_S$ if and only if $uv \notin E(G)$.
Therefore, we see that
\[
H_S = \conj{G[V_S]} = \conj{G[\conj{N[S]}]}.
\]
With the description of $H_S$, we are ready to prove the second item in \cref{ss:proofIS}.

\begin{lemma} \label{lem:IS-connected}
Let $G = (V, E)$ be a graph with maximum degree $\Delta$.
Suppose $k \leq \frac{|V|}{\Delta+1}$.
For any $S \in I_{G,k}$ with $|S| \leq k-2$,
the random walk matrix $\Emm_S$ of the graph $G_S$ of the link $(I_{G,k})_S$ satisfies $\lambda_2(\Emm_S) < 1$. 
\end{lemma}
\begin{proof}
Note that $\lambda_2(\Emm_S) < 1$ if and only if the underlying support graph $H_S$ of $G_S$ is connected, so we focus on proving the latter.
To prove that $H_S$ is connected, we prove the stronger claim that every two vertices $u,v \in H_S$ has a path of length at most two.
If $uv$ is an edge in $H_S$, then there is a path of length one.
Suppose $uv$ is not an edge in $H_S$.
Then $uv$ is an edge in $G$.
Since $G$ is of maximum degree $\Delta$, it implies that 
$|N_G[\{u,v\}]| \leq (\deg_G(u)+1) + (\deg_G(v)+1) - 2 \leq 2\Delta$,
and also
\[
|V_S| = |V| - |N_G[S]| \geq |V| - |S| \cdot (\Delta+1) \geq 2\Delta+2,
\]
    where we use the assumptions that $|S| \leq k-2 \leq \frac{|V|}{\Delta+1} - 2$ in the last inequality.
So, there must be some vertex $w$ such that $w \in V_S \setminus N_G[\{u,v\}]$.
This implies that $wu \notin E(G)$ and $wv \notin E(G)$,
and thus $wu \in E(H_S)$ and $wv \in E(H_S)$ and so there is a path of length two connecting $u$ and $v$ in $H_S$.
\end{proof}

\subsubsection{Proof of \cref{lem:IS-top-link}} \label{ss:IS-top-link}

We observe that $G_S$ is an unweighted graph for $S$ with $|S|=k-2$ when the distribution on $I_{G,k}(k-1)$ is the uniform distribution.
Therefore, $G_S$ is simply a scaled version of $H_S$, and the random walk matrix $\Emm_S$ of $G_S$ is the same as the random walk matrix of $H_S$.
To bound the second eigenvalue, we will use some simple interlacing arguments.
We need the stronger assumption that $k \leq \frac{|V(G)|}{\Delta +
|\lambda_{\min}(\Aye_G|)}$ in the proof of the following lemma. (Note that for
any unweighted graph $G$, we have $|\lambda_{\min}(\Aye_G)| \ge 1$.)

\begin{lemma}\label{lem:IS-top-link}
Let $G = (V, E)$ be a graph with maximum degree $\Delta$.
Suppose $k \leq |V|/(\Delta + |\lambda_{\min}(\Aye_G)|)$.
For any $S \in I_{G,k}$ with $|S| = k-2$,
the random walk matrix $\Emm_S$ of the graph $G_S$ of the link $(I_{G,k})_S$ satisfies $\lambda_2(\Emm_S) \leq \frac{1}{k}$. 
\end{lemma}
\begin{proof}[Proof of \cref{lem:IS-top-link}]
Recall that for $S$ with $|S|=k-2$,
the random walk matrix $\Emm_S$ of $G_S$ is the same as the random walk matrix of $H_S$, and so we will focus on the latter.
Let $\Dee_H$ be diagonal degree matrix of $H_S$.
As argued above, the random walk matrix $\Emm_S$ of $G_S$ is equal to 
$\Emm_S = \Dee_H^{-1} \Aye_H$.
We can write the adjacency matrix $\Aye_H$ of $H_S$ as
\[
\Aye_H = \one \one^\top- \Ide- \Aye_{G[\conj{N[S]}]},
\]
where $\Aye_{G[\conj{N[S]}]}$ is the adjacency matrix of $G[\conj{N[S]}]$.
By Weyl's interlacing theorem,
\begin{align}
        \lambda_2(\Emm_S) &~\le~\lambda_2(\Dee_H^{-1/2} \one\one^\top
        \Dee_H^{-1/2}) + \lambda_1(\Dee_H^{-1/2}(-\Aye_{G[\conj{N[S]}]} - \Ide)
        \Dee_H^{-1/2}),&&\textrm{ (by Weyl Interlacing \cref{thm:weyl-int})}\notag\\
        &~=~\lambda_1(\Dee_H^{-1/2}(-\Aye_{G[\conj{N[S]}]} - \Ide) 
        \Dee_H^{-1/2}), &&\textrm{ (since $\Dee_H^{-1/2} \one\one^\top \Dee_H^{-1/2}$
        is of rank } 1)\notag
\\
        &~\le~\norm*{\Dee_H^{-1}} \cdot \lambda_1(-\Aye_{G[\conj{N[S]}]} -
        \Ide),&&\textrm{ (by the \ref{eq:var-form})}\label{eq:refer-later}
\\
        &~\le~\norm*{\Dee_H^{-1}} \cdot (|\lambda_{\min}(\Aye_{G[\conj{N[S]}]})| -
        1),&&\textrm{ (by using $\lambda_1(-\Aye_{G[\conj{N[S]}]}) = -\lambda_{\min}(\Aye_{G[\conj{N[S]}]})$)}
        \notag
\\
        &~\le~\norm*{\Dee_H^{-1}} \cdot (|\lambda_{\min}(\Aye_{G})| -
        1),&&\textrm{ (by using Cauchy-Interlacing \cref{thm:cauchy-int})}\notag
\end{align}
For \cref{eq:refer-later}, we have used the \ref{eq:var-form}
$\lambda_1(\Why) = \max\set*{\langle \eff, \Why \eff\rangle : \eff \in \RR^V, \norm{\eff}= 1}$ in the following way: 
Let $\Why = -\Aye_L - \Ide$ and $\gee$ be an unit top-eigenvecor of $\Dee_H^{-1/2} \Why \Dee_H^{-1/2}$, i.e.~$\norm{\gee} = 1$ and $\langle \gee, \Dee_H^{-1/2} \Why \Dee_H^{-1/2} \gee\rangle = \lambda_1(\Dee_H^{-1/2}\Why \Dee_H^{-1/2})$. 
Then, 
\[ \lambda_1(\Why) 
\ge \left\langle \frac{\Dee_H^{-1/2} \gee}{\norm{\Dee_H^{-1/2} \gee}}, \Why
        \frac{\Dee_H^{-1/2} \gee}{\norm{\Dee_H^{-1/2} \gee}}\right\rangle
= \frac{\lambda_1(\Dee_H^{-1/2} \Why\Dee_H^{-1/2}) }{\norm{\Dee_H^{-1/2}}^2}
\ge \frac{\lambda_1(\Dee_H^{-1/2} \Why\Dee_H^{-1/2})}{\norm{\Dee_H^{-1}}}.
\]

It remains to bound $\norm{\Dee_H^{-1}} = (\min_v \deg_{H_S}(v))^{-1}$.
As $H_S = \conj{G[\conj{N[S]}]} = \conj{G[V-N[S]]}$,
\[ 
\deg_{H_S}(v) 
= |V| - |N[S]| -(\deg_{G[\conj{N[S]}]}(v) + 1)
\geq |V| - (\Delta+1)(|S|+1),
\]
where the last inequality uses that $|N[S]| \leq |S| \cdot (\Delta+1)$ and $\deg_{G[\conj{N[S]}]}(v) \leq \deg_G(v) \leq \Delta$.
    Therefore, using our bound $\lambda_2(\Emm_S) \le \norm*{\Dee_H^{-1}} \cdot (|\lambda_{\min}(\Aye_{G})| - 1)$, we obtain
\[ 
\lambda_2(\Emm_S) 
\le \frac{|\lambda_{\min}(\Aye_G)| - 1}{|V| - (\Delta + 1) \cdot (|S| + 1)}
= \frac{|\lambda_{\min}(\Aye_G)| - 1}{|V| - (\Delta + 1) \cdot (k - 1)},
\]
where we use $|S| = k- 2$. 
Finally, plugging in the assumption
\[
k \leq \frac{|V|}{\Delta + |\lambda_{\min}(\Aye_G)|}
\quad \implies \quad 
\lambda_2(\Emm_S) \le \frac{1}{k}.
\]
\end{proof}

\subsection{Matroid Intersection}\label{ss:mint-main}

A matroid $M = (E, \Ii)$ on the ground set $E$ with the set of independent sets
${\Ii} \subset 2^E$ is a combinatorial object satisfying the following properties:
\begin{itemize}
    \item (containment property) if $S \in \Ii$ and $T \subset S$, then $T \in \Ii$,
    \item (extension property) if $S, T \in \Ii$ such that $|S| > |T|$ then
        there is some $x \in S\backslash T$ such that $\set{x} \cup T \in \Ii$.
\end{itemize}
A partition matroid is the special case where the ground set $E$ is partitioned into disjoint blocks $B_1, \ldots, B_l \subseteq E$ with parameters $0 \leq d_i \leq |B_i|$ for $1 \leq i \leq l$,
and a subset $S$ is in $\Ii$ if and only if $|S \cap B_i| \leq d_i$ for $1 \leq i \leq l$.

The intersection of two matroids $M_1 = (E, \Ii_1)$ and $M_2 = (E, \Ii)$ over the same ground set $E$ can be used to formulate various interesting combinatorial optimization problems~\cite{Schrijver03}.
We are interested in the problem of sampling a uniform random common
independent set of size $k$, i.e.~a random subset $F \in \Ii_1 \cap \Ii_2$ with $|F|=k$.
We will analyze a natural Markov chain for the problem by analying the down-up walk of a corresponding simplicial complex.

Define the $(k-1)$-dimensional matroid intersection complex $C_{M_1, M_2, k}$ of $M_1 = (E, \Ii_1)$ and $M_2 = (E, \Ii_2)$ as
\[ 
C_{M_1, M_2, k} = \set*{ S \in \Ii_1 \cap \Ii_2 : |S| \le k},
\]
the complex consisting of all common independent sets of both matroids
containing at most $k$ elements. 
We endow $C_{M_1, M_2, k}(k-1)$ with the uniform distribution $\Pi_{k-1}$ on the common independent sets $S \in \Ii_1 \cap \Ii_2$ with $|S| = k$. 
We write $C_{M_1, M_2, k}$ for the weighted simplicial complex $(C_{M_1, M_2, k}, \Pi_{k-1})$.

The $(k-1)$-th down-up walk $\DownW_{k-1}$ on $C_{M_1, M_2, k}$ corresponds to the following natural Markov chain to sample common independent sets of size $k$.
Initially, the random walk starts from a common independent set $S_1$ of size $k$. 
In each step $t \geq 1$, we choose a uniform random element $i \in S_t$ and delete $i$ from $S_t$, and set $S_{t+1}$ to be a uniform random common independent set of size $k$ that contains $S_t\setminus \{i\}$.
The stationary distribution of $\DownW_{k-1}$ is the uniform distribution $\Pi_{k-1}$; see \cref{eq:stationary}.

The main result in this subsection is the following upper bound on the second eigenvalue of $\DownW_{k-1}$.

\matroid*

\subsubsection{Proof of \cref{t:matroid}} \label{sss:proof}

The plan is to use \cref{cor:main-cooked} to prove \cref{t:matroid}.
To apply \cref{cor:main-cooked}, we need to prove that:
\begin{enumerate}
\item $C_{M_1, M_2, k}$ is a pure simplicial complex.  
This is a simple proof in \cref{prop:mint-pure}.
\item For each $S \in C_{M_1, M_2, k}$ with $|S| \leq k-2$, 
the random walk matrix $\Emm_S$ of the graph $G_S$ of the link $(C_{M_1,M_2,k})_S$ satisfies $\lambda_2(\Emm_S) < 1$.
This is proved in \cref{lem:connected}, showing that the underlying graph of $G_S$ is the complement of the line graph of a bipartite graph.
\item For each $S \in C_{M_1, M_2, k}$ with $|S| = k-2$,
the random walk matrix $\Emm_S$ of the graph $G_S$ satisfies $\lambda_2(\Emm_S) \leq 1/k$.
This is proved in \cref{lem:top-bound-mint}, using the fact that the minimum eigenvalue of the adjacency matrix of the line graph of a simple graph is at least $-2$.
\end{enumerate}
Assuming the three items are proven, \cref{t:matroid} follows immediately from \cref{cor:main-cooked}.

It remains to prove the three items.
We will prove the second item in \cref{ss:connected}
and the third item in \cref{ss:top-link}.
We note that the first two items hold for any two matroids, and we only use the additional assumptions for the third item.
The following is a simple proof for the first item.

\begin{claim}\label{prop:mint-pure}
Let $M_1 = (E, \Ii_1)$ and $M_2 = (E, \Ii_2)$ be two matroids with a common independent set $T \in \Ii_1 \cap \Ii_2$ of size $|T|=r$. 
Any common independent set $S \in \Ii_1 \cap \Ii_2$ with $|S| < r/2$ is
 contained in a larger common independent set. 
In particular, this implies that the simplicial complex $C_{M_1, M_2, k}$ is a pure simplicial complex as long as $k \le r/2$.
\end{claim}
\begin{proof}
By the extension property of matroids, 
there is a subset $T_1 \subset T$ with $|T_1| \geq r-|S|$ such that $S \cup \set{x}$ is an independent set in $\Ii_1$ for any $x \in T_1$.
Similarly, there is a subset $T_2 \subset T$ with $|T_2| \geq r-|S|$ such that $S \cup \set{y}$ is an independent set in $\Ii_2$ for any $y \in T_2$.
As $|S| < r/2$, this implies that $T_1 \cap T_2 \neq \emptyset$, and $S \cup \{z\}$ is a larger independent set that contains $S$ for any $z \in T_1 \cap T_2$.
\end{proof}

\subsubsection{Proof of \cref{lem:connected}} \label{ss:connected}

Let $H_S=(E_S,F_S)$ be the underlying support graph of $G_S$ of the link $(C_{M_1,M_2,k})_S$, that is, $H_S$ is $G_S$ without edge weights.
The vertex set of $H_S$ is $E_S = \{x \in E \mid S \cup \set{x} \in \Ii_1 \cap \Ii_2\}$ 
and the edge set of $H_S$ is $F_S = \{ \set{x,y} \mid x,y \in E {\rm~and~} S \cup \{x,y\} \in \Ii_1 \cap \Ii_2\}$.
Let $\Emm_S$ be the random walk matrix of $G_S$ as defined in \cref{ss:local-spectral}.
It is a basic fact in spectral graph theory that $\lambda_2(\Emm_S) < 1$ if and only if $H_S$ is connected.

We will see that $H_S$ is the complement of the line graph of a bipartite graph $B$.
To define the bipartite graph $B$, we first introduce the matroid partition property (see e.g.~\cite{AnariLOV18}).
The matroid partition property says that there is  a partition ${\cal P} := \{P_1, \ldots, P_p\}$ of the vertex set $E_S$ (i.e. $\bigcup_{i=1}^p P_i = E_S$ and $P_i \cap P_j = \emptyset$ for $i \neq j$) with the property that for any $x,y \in E_S$,
\[
S \cup \set{x,y} \notin \Ii_1 \quad \iff \quad x,y \in P_i {\rm~for~some~} 1 \leq i \leq p.
\]
In words, there is a partition ${\cal P}$ of the vertex set $E_S$ such that two elements $x,y$ in $E_S$ can be added to $S$ to form an independent set in the first matroid $M_1$ if and only if $x,y$ do not belong to the same class of the partition ${\cal P}$.
Similarly, there is a partition ${\cal Q} := \{Q_1, \ldots, Q_q\}$ of the vertex set $E_S$ such that for any two elements $x,y \in E_S$, we have $S \cup \set{x,y} \notin \Ii_2$ if and only if $x,y \in Q_i$ for some $1 \leq i \leq q$.

We use the partitions ${\cal P}$ and ${\cal Q}$ to define the bipartite graph $B$ as follows.
The vertex set of $B$ is $P \sqcup Q$,
where we create a vertex $i \in P$ in $B$ for each $P_i$ in ${\cal P}$,
and we create a vertex $j \in Q$ in $B$ for each $Q_j$ in ${\cal Q}$.
Each edge in $B$ corresponds to an element in $E_S$.
For each element $x \in E_S$, we create the edge $e_x=ij$ in $B$ if and only if $x \in P_i$ and $x \in Q_j$.
Note that the edge $e_x$ for $x \in E_S$ is well-defined by the matroid partition property.
By construction, it should be clear that the biparite graph $B$ satisfies the following important property:
\begin{equation} \label{eq:iff}
e_x {\rm~and~} e_y {\rm~do~not~share~a~vertex~in~} B
\quad \iff \quad
S \cup \{x,y\} \in \Ii_1 \cap \Ii_2 
\quad \iff \quad
\{x,y\} \in F_S.
\end{equation}
Recall that the line graph $L(B)$ of a graph $B$ is defined as follows:
the vertex set of $L(B)$ is the edge set of $B$,
and two vertices in $L(B)$ have an edge if and only if the corresponding edges in $B$ share an endpoint.
Let $\overline{L(B)}$ be the complement of $L(B)$ where $\overline{L(B)}$ and $L(B)$ have the same vertex set and two vertices in $\overline{L(B)}$ have an edge if and only if the corresponding vertices in $L(B)$ do not have an edge.
Then, we see from \cref{eq:iff} that
\begin{equation} \label{eq:H}
H_S = \overline{L(B)}.
\end{equation}
Using the bipartite graph $B$, it is easy to show the second item in \cref{sss:proof}.

\begin{lemma} \label{lem:connected}
Let $M_1 = (E, \Ii_1)$ and $M_2 = (E, \Ii_2)$ be two matroids with a common independent set $T \in \Ii_1 \cap \Ii_2$ of size $|T|=r$. 
Suppose $k < r/2 - 1$.
For any $S \in C_{M_1,M_2,k}$ with $|S| \leq k-2$,
the random walk matrix $\Emm_S$ of the graph $G_S$ of the link $(C_{M_1,M_2,k})_S$ satisfies $\lambda_2(\Emm_S) < 1$. 
\end{lemma}
\begin{proof}
It is well known that $\lambda_2(\Emm_S) < 1$ if and only if the underlying support graph $H_S$ of $G_S$ is connected, so we focus on proving the latter.
Since $|S| \leq k-2 < r/2 - 3$, 
it follows from \cref{prop:mint-pure} that there are four elements $a,b,c,d \in E$ such that $S \cup \{a,b,c,d\} \in \Ii_1 \cap \Ii_2$.
In the bipartite graph $B$ in \cref{eq:H},
the four elements $a,b,c,d$ correspond to four vertex-disjoint edges $e_a,e_b,e_c,e_d$ in $B$ by \cref{eq:iff}.
To prove that $H_S$ is connected, we prove the stronger claim that every two vertices $u,v \in H_S$ has a path of length at most two.
If $uv$ is an edge in $H_S$, then there is a path of length one.
Suppose $uv$ is not an edge in $H_S$.
Then $e_u$ and $e_v$ shares a vertex in $B$ and so they span at most three vertices in $B$.
This implies that $e_u \cup e_v$ cannot intersect all four (vertex-disjoint) edges $e_a,e_b,e_c,e_d$.
So there must be an edge, say $e_a$, which is vertex-disjoint from both $e_u$ and $e_v$.
Then $u$-$a$-$v$ is path of length two in $H_S$ by \cref{eq:iff}, which completes the proof.
\end{proof}

\subsubsection{Proof of \cref{lem:top-bound-mint}} \label{ss:top-link}

For the third term, we need to prove that for each $S \in C_{M_1,M_2,k}$ with $|S|=k-2$, the random walk matrix $\Emm_S$ of the graph $G_S$ satisfies $\lambda_2(\Emm_S) \leq \frac{1}{k}$.
We use the additional assumptions for the following property.

\begin{claim} \label{c:simple}
If $M_1$ and $M_2$ are two partition matroids and there are no two elements $x,y$ such that $x,y$ belongs to the same block in $M_1$ and also the same block in $M_2$, then \cref{eq:H} holds with the property that the bipartite graph $B$ is a simple graph.
\end{claim}

Observe that $G_S$ is an unweighted graph for $S$ with $|S|=k-2$ when the distribution on $C_{M_1,M_2,k}(k-1)$ is the uniform distribution (i.e. the distribution on the common independent sets of size $k$ is the uniform distribution).
This is because when $|S|=k-2$, for any $x,y \in E$, either $S \cup \{x,y\}$ is contained in exactly one or zero set of size $k$ in $C_{M_1,M_2,k}$, and each set of size $k$ is assigned the same weight in the uniform distribution (more formally see \cref{eq:link-def} for the definition of the weight).
Therefore, $G_S$ is simply a scaled version of $H_S$, and the random walk matrix $\Emm_S$ of $G_S$ is the same as the random walk matrix of $H_S$.

\begin{fact}\label{fac:lg-eigbd}
Let $G = (V, E)$ be any simple graph and $\Aye_{L(G)}$ be the adjacency matrix of the line graph of $G$.
It holds that $\lambda_{\min}(\Aye_{L(G)}) \ge -2$.
\end{fact}
\begin{proof}
Define $\Bee \in \RR^{E \times V}$ to be the edge-vertex incidence matrix of $G = (V, E)$, i.e.~$\Bee(e, v) = 1[v \in e]$.
Observe that
\[ 
2\Ide + \Aye_{L(G)} = \Bee\Bee^\top 
\quad \implies \quad
\lambda_{\min}\left(2\Ide + \Aye_{L(G)}\right) = \lambda_{\min}\left(\Bee\Bee^\top\right) \ge 0,
\]
as $\Bee\Bee^\top$ is a positive semidefinite matrix.
This implies that $\lambda_{\min}(\Aye_{L(G)}) \ge -2$.
\end{proof}

We are ready to bound the second eigenvalue of $\Emm_S$.
We need the stronger assumption that $k \leq \frac{r}{3}$ in the proof of the following lemma.

\begin{restatable}{lemma}{topboundmint}\label{lem:top-bound-mint}
Let $M_1 = (E, \Ii_1)$ and $M_2 = (E, \Ii_2)$ be two partition matroids with a common independent set $T \in \Ii_1 \cap \Ii_2$ of size $|T|=r$ and there are no two elements belonging to the same block in both matroids. 
Suppose $k \leq r/3$.
For any $S \in C_{M_1,M_2,k}$ with $|S| = k-2$,
the random walk matrix $\Emm_S$ of the graph $G_S$ of the link $(C_{M_1,M_2,k})_S$ satisfies $\lambda_2(\Emm_S) \leq 1/k$. 
\end{restatable}
\begin{proof}
Recall that for $S$ with $|S|=k-2$,
the random walk matrix $\Emm_S$ of $G_S$ is the same as the random walk matrix of $H_S$, and so we will focus on the latter.
By \cref{c:simple}, $H_S$ is the complement of the line graph of a simple graph,
and so we can write the adjacency matrix $\Aye_H$ of $H_S$ as
\[
\Aye_H = \one \one^\top- \Ide- \Aye_L, 
\]
where $\Aye_L$ is the adjacency matrix of the line graph $L(B)$ in \cref{eq:H}.
Let $\Dee_H$ be diagonal degree matrix of $H_S$.
As argued above, the random walk matrix $\Emm_S$ of $G_S$ is equal to 
$\Emm_S = \Dee_H^{-1} \Aye_H$.
Using that $\Dee_H^{-1}$ and $\Dee_H^{-1/2}\Aye_H\Dee_H^{-1/2}$ are similar matrices and have the same spectrum, we have
\begin{align*}
\lambda_2(\Emm_S) &~=~ \lambda_2(\Dee_H^{-1/2} \Aye_H \Dee_H^{-1/2}),
\\
&~=~\lambda_2(\Dee_H^{-1/2}(\one\one^\top - \Aye_{L} -\Ide)\Dee_H^{-1/2}),
\\
&~=~\lambda_2\parens*{\Dee_H^{-1/2}\one\one^\top\Dee_H^{-1/2} +
        \Dee_H^{-1/2}(-\Aye_{L} - \Ide)\Dee_H^{-1/2}}.
\end{align*}
Using the Weyl Interlacing \cref{thm:weyl-int},
\begin{align}
        \lambda_2(\Emm_S) &~\le~\lambda_2(\Dee_H^{-1/2} \one\one^\top
        \Dee_H^{-1/2}) + \lambda_1(\Dee_H^{-1/2}(-\Aye_{L} - \Ide)
        \Dee_H^{-1/2}),&&\textrm{ (by Weyl Interlacing \cref{thm:weyl-int})}\notag\\
        &~=~\lambda_1(\Dee_H^{-1/2}(-\Aye_{L} - \Ide) 
        \Dee_H^{-1/2}), &&\textrm{ (since $\Dee_H^{-1/2} \one\one^\top \Dee_H^{-1/2}$
        is of rank } 1)\notag
\\
        &~\le~\norm*{\Dee_H^{-1}} \cdot \lambda_1(-\Aye_{L} -
        \Ide),&&\textrm{ (by the \ref{eq:var-form})}\label{eq:refer-later2}
\\
        &~\le~\norm*{\Dee_H^{-1}} \cdot (|\lambda_{\min}(\Aye_L)| -
        1),&&\textrm{ (by using $\lambda_1(-\Aye_L) = -\lambda_{\min}(\Aye_L)$)
        }\notag
\\
        &~=~\norm*{\Dee_H^{-1}}.&&\textrm{ (by using \cref{fac:lg-eigbd})}\label{eq:last-mint}
\end{align}
For \cref{eq:refer-later2}, we have used the \ref{eq:var-form}
$\lambda_1(\Why) = \max\set*{\langle \eff, \Why \eff\rangle : \eff \in \RR^V, \norm{\eff}= 1}$ in the following way: 
Let $\Why = -\Aye_L - \Ide$ and $\gee$ be an unit top-eigenvecor of $\Dee_H^{-1/2} \Why \Dee_H^{-1/2}$, i.e.~$\norm{\gee} = 1$ and $\langle \gee, \Dee_H^{-1/2} \Why \Dee_H^{-1/2} \gee\rangle = \lambda_1(\Dee_H^{-1/2}\Why \Dee_H^{-1/2})$. 
Then, 
\[ \lambda_1(\Why) 
\ge \left\langle \frac{\Dee_H^{-1/2} \gee}{\norm{\Dee_H^{-1/2} \gee}}, \Why
        \frac{\Dee_H^{-1/2} \gee}{\norm{\Dee_H^{-1/2} \gee}}\right\rangle
= \frac{\lambda_1(\Dee_H^{-1/2} \Why\Dee_H^{-1/2}) }{\norm{\Dee_H^{-1/2}g}^2}
\ge \frac{\lambda_1(\Dee_H^{-1/2} \Why\Dee_H^{-1/2})}{\norm{\Dee_H^{-1}}}.
\]

It remains to bound $\norm{\Dee_H^{-1}} = (\min_x \deg_{H_S}(x))^{-1}$.
By the definition of $H_S$, the degree $\deg_{H_S}(x)$ of $x \in E$ is equal to the number of elements $y \in E \setminus (S \cup \{x\})$ such that $S \cup \{x,y\} \in \Ii_1 \cap \Ii_2$.
By our assumption, there is a common independent set $T \in \Ii_1 \cap \Ii_2$ of size $r$. 
Since $|S \cup \{x\}| = k-1$, by the extension property of the first matroid $M_1$, there are at least $r-k+1$ elements $y \in T$ such that $S \cup \{x,y\} \in \Ii_1$.
Similarly, there are at least $r-k+1$ elements $y \in T$ such that $S \cup \{x,y\} \in \Ii_2$.
Therefore, there are at least $r-2k+2$ elements $y \in T$ such that $S \cup \{x,y\} \in \Ii_1 \cap \Ii_2$.
This implies that for any $x \in V(H_S)$, 
\[
\deg_{H_S}(x) \geq r - 2k + 2 \geq k 
\quad \implies \quad
\lambda_2(\Emm_S) \leq \norm{\Dee_H^{-1}} \leq \frac{1}{k},
\]
where we use the assumption that $k \leq \frac{r}{3}$.
\end{proof}

\section*{Acknowledgements}
We thank Kuikui Liu, Shayan Oveis-Gharan, Akshay Ramachandran, and Hong Zhou for many useful discussions.
\bibliographystyle{alpha}
\bibliography{comp2}

\newcommand{\etalchar}[1]{$^{#1}$}
\begin{thebibliography}{DHK{\etalchar{+}}19}

\bibitem[AB06]{AharoniB06}
Ron Aharoni and Eli Berger.
\newblock The intersection of a matroid and a simplicial complex.
\newblock {\em Transactions of the American Mathematical Society},
  358(11):4895--4917, 2006.

\bibitem[AHK18]{AdirpasitoHK18}
Karim Adiprasito, June Huh, and Eric Katz.
\newblock Hodge theory for combinatorial geometries.
\newblock {\em Annals of Mathematics}, 188(2):381--452, 2018.

\bibitem[AJQ{\etalchar{+}}20]{AlevJQST19}
Vedat~Levi Alev, {Fernando Granha} Jeronimo, Dylan Quintana, Shashank
  Srivastava, and Madhur Tulsiani.
\newblock List decoding of direct sum codes.
\newblock In {\em SODA}, 2020.

\bibitem[AJT19]{AlevJT19}
Vedat~Levi Alev, {Fernando Granha} Jeronimo, and Madhur Tulsiani.
\newblock Approximating constraint satisfaction problems on high dimensional
  expanders.
\newblock In {\em FOCS}, 2019.

\bibitem[Ald83]{Aldous83}
David Aldous.
\newblock Random walks on finite groups and rapidly mixing markov chains.
\newblock In {\em S{\'e}minaire de Probabilit{\'e}s XVII 1981/82}, pages
  243--297. Springer, 1983.

\bibitem[ALO]{AnariLO19}
Nima Anari, Kuikui Liu, and Shayan {Oveis Gharan}.
\newblock Spectral Independence in High-Dimensional Expanders and Applications
  to the Hardcore Model (Manuscript).

\bibitem[ALOV19]{AnariLOV18}
Nima Anari, Kuikui Liu, Shayan {Oveis Gharan}, and Cynthia Vinzant.
\newblock Log-concave polynomials ii: high-dimensional walks and an fpras for
  counting bases of a matroid.
\newblock In {\em {STOC}}, pages 1--12. ACM, 2019.

\bibitem[AM85]{AlonMilman}
Noga Alon and V.~D. Milman.
\newblock lambda\({}_{\mbox{1}}\), isoperimetric inequalities for graphs, and
  superconcentrators.
\newblock {\em J. Comb. Theory, Ser. {B}}, 38(1):73--88, 1985.

\bibitem[BD97]{BubleyD97}
Russ Bubley and Martin~E. Dyer.
\newblock Path coupling: {A} technique for proving rapid mixing in markov
  chains.
\newblock In {\em {FOCS}}, pages 223--231, 1997.

\bibitem[BH19]{BrandenH19}
Petter Br{\"a}nd{\'e}n and June Huh.
\newblock Lorentzian polynomials.
\newblock {\em arXiv preprint arXiv:1902.03719}, 2019.

\bibitem[Bha13]{Bhatia2013}
Rajendra Bhatia.
\newblock {\em Matrix analysis}, volume 169.
\newblock Springer Science \& Business Media, 2013.

\bibitem[BT06]{BobkovT06}
Sergey~G Bobkov and Prasad Tetali.
\newblock Modified logarithmic sobolev inequalities in discrete settings.
\newblock {\em Journal of Theoretical Probability}, 19(2):289--336, 2006.

\bibitem[CGM19]{CryanGM19}
Mary Cryan, Heng Guo, and Giorgos Mousa.
\newblock Modified log-sobolev inequalities for strongly log-concave
  distributions.
\newblock {\em CoRR}, abs/1903.06081, 2019.

\bibitem[Che70]{Cheeger}
Jeff Cheeger.
\newblock A lower bound for the smallest eigenvalue of the laplacian.
\newblock {\em Problems in analysis}, pages 195--199, 1970.

\bibitem[DD19]{DiksteinD19}
Yotam Dikstein and Irit Dinur.
\newblock Agreement testing theorems on layered set systems.
\newblock 2019.

\bibitem[DDFH18]{DiksteinDFH18}
Yotam Dikstein, Irit Dinur, Yuval Filmus, and Prahladh Harsha.
\newblock Boolean function analysis on high-dimensional expanders.
\newblock In {\em {APPROX/RANDOM}}, pages 38:1--38:20, 2018.

\bibitem[DFK91]{DyerFK91}
Martin Dyer, Alan Frieze, and Ravi Kannan.
\newblock A random polynomial-time algorithm for approximating the volume of
  convex bodies.
\newblock {\em Journal of the ACM (JACM)}, 38(1):1--17, 1991.

\bibitem[DHK{\etalchar{+}}19]{DinurHKNT19}
Irit Dinur, Prahladh Harsha, Tali Kaufman, Inbal~Livni Navon, and Amnon
  Ta{-}Shma.
\newblock List decoding with double samplers.
\newblock In {\em {SODA}}, pages 2134--2153, 2019.

\bibitem[DK17]{DinurK17}
Irit Dinur and Tali Kaufman.
\newblock High dimensional expanders imply agreement expanders.
\newblock In {\em {FOCS}}, pages 974--985, 2017.

\bibitem[DKW16]{DotterrerKW16}
Dominic Dotterrer, Tali Kaufman, and Uli Wagner.
\newblock On expansion and topological overlap.
\newblock In {\em {SOCG}}, pages 35:1--35:10, 2016.

\bibitem[DSC{\etalchar{+}}96]{DiaconisSC96}
Persi Diaconis, Laurent Saloff-Coste, et~al.
\newblock Logarithmic sobolev inequalities for finite markov chains.
\newblock {\em The Annals of Applied Probability}, 6(3):695--750, 1996.

\bibitem[EK16]{EvraK16}
Shai Evra and Tali Kaufman.
\newblock Bounded degree cosystolic expanders of every dimension.
\newblock In {\em {STOC}}, pages 36--48, 2016.

\bibitem[FGL{\etalchar{+}}11]{FoxGLNP11}
Jacob Fox, Mikhail Gromov, Vincent Lafforgue, Assaf Naor, and J{\'{a}}nos Pach.
\newblock Overlap properties of geometric expanders.
\newblock In {\em {SODA}}, pages 1188--1197, 2011.

\bibitem[FM92]{FederM92}
Tom Feder and Milena Mihail.
\newblock Balanced matroids.
\newblock In {\em Proceedings of the Twenty Fourth Annual ACM Symposium on the
  Theory of Computing}, pages 26--38. Citeseer, 1992.

\bibitem[God]{GodsilWP}
Chris Godsil.
\newblock Mathematics stackexchange url:
  \url{https://math.stackexchange.com/questions/2527245/what-graph-have-exactly-one-positive-eigenvalue},.

\bibitem[Gro10]{Gromov10}
Mikhail Gromov.
\newblock Singularities, expanders and topology of maps. part 2: From
  combinatorics to topology via algebraic isoperimetry.
\newblock {\em Geometric and Functional Analysis}, 20(2):416--526, 2010.

\bibitem[Hay06]{Hayes06}
Thomas~P Hayes.
\newblock A simple condition implying rapid mixing of single-site dynamics on
  spin systems.
\newblock In {\em FOCS}, pages 39--46. IEEE, 2006.

\bibitem[HLW06]{HooryLW06}
Shlomo Hoory, Nathan Linial, and Avi Wigderson.
\newblock Expander graphs and their applications.
\newblock {\em Bulletin of the American Mathematical Society}, 43(4):439--561,
  2006.

\bibitem[HW17]{HuhW17}
June Huh and Botong Wang.
\newblock Enumeration of points, lines, planes, etc.
\newblock {\em Acta Mathematica}, 218(2):297--317, 2017.

\bibitem[Jer95]{Jerrum95}
Mark Jerrum.
\newblock A very simple algorithm for estimating the number of k-colorings of a
  low-degree graph.
\newblock {\em Random Struct. Algorithms}, 7(2):157--166, 1995.

\bibitem[JS89]{JerrumS89}
Mark Jerrum and Alistair Sinclair.
\newblock Approximating the permanent.
\newblock {\em {SIAM} J. Comput.}, 18(6):1149--1178, 1989.

\bibitem[JST{\etalchar{+}}04]{JerrumSTV04}
Mark Jerrum, Jung-Bae Son, Prasad Tetali, Eric Vigoda, et~al.
\newblock Elementary bounds on poincar{\'e} and log-sobolev constants for
  decomposable markov chains.
\newblock {\em The Annals of Applied Probability}, 14(4):1741--1765, 2004.

\bibitem[JSV04]{JerrumSV04}
Mark Jerrum, Alistair Sinclair, and Eric Vigoda.
\newblock A polynomial-time approximation algorithm for the permanent of a
  matrix with nonnegative entries.
\newblock {\em J. {ACM}}, 51(4):671--697, 2004.

\bibitem[JVV86]{JerrumVV86}
Mark~R Jerrum, Leslie~G Valiant, and Vijay~V Vazirani.
\newblock Random generation of combinatorial structures from a uniform
  distribution.
\newblock {\em Theoretical Computer Science}, 43:169--188, 1986.

\bibitem[KKL16]{KaufmanKL16}
Tali Kaufman, David Kazhdan, and Alexander Lubotzky.
\newblock Isoperimetric inequalities for ramanujan complexes and topological
  expanders.
\newblock {\em Geometric and Functional Analysis}, 26(1):250--287, 2016.

\bibitem[KL14]{KaufmanL14}
Tali Kaufman and Alexander Lubotzky.
\newblock High dimensional expanders and property testing.
\newblock In {\em (ITCS)}, pages 501--506, 2014.

\bibitem[KM16]{KaufmanM16a}
Tali Kaufman and David Mass.
\newblock Walking on the edge and cosystolic expansion.
\newblock {\em CoRR}, abs/1606.01844, 2016.

\bibitem[KM17]{KaufmanM17}
Tali Kaufman and David Mass.
\newblock High dimensional random walks and colorful expansion.
\newblock In {\em {ITCS}}, pages 4:1--4:27, 2017.

\bibitem[KO18]{KaufmanO18}
Tali Kaufman and Izhar Oppenheim.
\newblock High order random walks: Beyond spectral gap.
\newblock In {\em {APPROX/RANDOM}}, pages 47:1--47:17, 2018.

\bibitem[LM06]{LinialM06}
Nathan Linial and Roy Meshulam.
\newblock Homological connectivity of random 2-complexes.
\newblock {\em Combinatorica}, 26(4):475--487, 2006.

\bibitem[LMY19]{LiuMY19}
Siqi Liu, Sidhanth Mohanty, and Elizabeth Yang.
\newblock High-dimensional expanders from expanders.
\newblock {\em CoRR}, abs/1907.10771, 2019.

\bibitem[LSV05]{LubotzkySV05}
Alexander Lubotzky, Beth Samuels, and Uzi Vishne.
\newblock Explicit constructions of ramanujan complexes of type.
\newblock {\em Eur. J. Comb.}, 26(6):965--993, 2005.

\bibitem[LV06]{LovaszV06}
L{\'a}szl{\'o} Lov{\'a}sz and Santosh Vempala.
\newblock Simulated annealing in convex bodies and an o*(n4) volume algorithm.
\newblock {\em Journal of Computer and System Sciences}, 72(2):392--417, 2006.

\bibitem[Mes01]{Meshulam06}
Roy Meshulam.
\newblock The clique complex and hypergraph matching.
\newblock {\em Combinatorica}, 21(1):89--94, 2001.

\bibitem[MT05]{MontenegroT05}
Ravi Montenegro and Prasad Tetali.
\newblock Mathematical aspects of mixing times in markov chains.
\newblock {\em Foundations and Trends in Theoretical Computer Science}, 1(3),
  2005.

\bibitem[MU05]{MitzenmacherUpfal05}
Michael Mitzenmacher and Eli Upfal.
\newblock {\em Probability and computing - randomized algorithms and
  probabilistic analysis}.
\newblock Cambridge University Press, 2005.

\bibitem[MV87]{MihailV87}
Milena Mihail and Umesh Vazirani.
\newblock On the expansion of 0/1 polytopes.
\newblock {\em Journal of Combinatorial Theory}, 1987.

\bibitem[MW09]{MeshulamW09}
R.~Meshulam and N.~Wallach.
\newblock Homological connectivity of random \emph{k}-dimensional complexes.
\newblock {\em Random Struct. Algorithms}, 34(3):408--417, 2009.

\bibitem[Opp18]{Oppenheim18}
Izhar Oppenheim.
\newblock Local spectral expansion approach to high dimensional expanders part
  {I:} descent of spectral gaps.
\newblock {\em Discrete {\&} Computational Geometry}, 59(2):293--330, 2018.

\bibitem[PRT16]{ParzanchevskiRT16}
Ori Parzanchevski, Ron Rosenthal, and Ran~J. Tessler.
\newblock Isoperimetric inequalities in simplicial complexes.
\newblock {\em Combinatorica}, 36(2):195--227, 2016.

\bibitem[Sch03]{Schrijver03}
Alexander Schrijver.
\newblock {\em Combinatorial optimization: polyhedra and efficiency},
  volume~24.
\newblock Springer Science \& Business Media, 2003.

\bibitem[Sin92]{Sinclair92}
Alistair Sinclair.
\newblock Improved bounds for mixing rates of markov chains and multicommodity
  flow.
\newblock {\em Combinatorics, probability and Computing}, 1(4):351--370, 1992.

\bibitem[Sin93]{Sinclair93}
Alistair Sinclair.
\newblock Algorithms for random generation and counting. progress in
  theoretical computer science, 1993.

\bibitem[Vig00]{Vigoda00}
Eric Vigoda.
\newblock Improved bounds for sampling colorings.
\newblock {\em Journal of Mathematical Physics}, 41(3):1555--1569, 2000.

\bibitem[Wei06]{Weitz06}
Dror Weitz.
\newblock Counting independent sets up to the tree threshold.
\newblock In {\em STOC}, pages 140--149. ACM, 2006.

\bibitem[WLP09]{WilmerLP09}
EL~Wilmer, David~A Levin, and Yuval Peres.
\newblock Markov chains and mixing times.
\newblock {\em American Mathematical Soc., Providence}, 2009.

\end{thebibliography}

\end{document}